\documentclass[mnsc,copyedit,arxiv]{informs3}

\OneAndAHalfSpacedXI

\usepackage{natbib}
 \bibpunct[, ]{(}{)}{,}{a}{}{,}%
 \def\bibfont{\small}%
\TheoremsNumberedThrough     

\EquationsNumberedThrough    

\MANUSCRIPTNO{MS-RMA-19-02050.R3}

\usepackage{latexsym}
\usepackage{graphicx}
\usepackage{amsmath,amssymb,enumerate,mathrsfs}
\usepackage{thmtools,thm-restate}
\usepackage{empheq}
\usepackage{xspace}
\usepackage{bm}
\usepackage{ifpdf}
\usepackage[usenames,dvipsnames,svgnames,table]{xcolor}
\usepackage{algorithm}
\usepackage[noend]{algpseudocode}
\usepackage{nicefrac}
\usepackage{wrapfig}
\usepackage{multicol}
\usepackage{paralist}

\definecolor{Darkblue}{rgb}{0,0,0.4}
\definecolor{Brown}{cmyk}{0,0.81,1.,0.60}
\definecolor{Purple}{cmyk}{0.45,0.86,0,0}
\definecolor{myblue}{RGB}{25,25,129}
\newcommand{\mydriver}{hypertex}
\ifpdf
 \renewcommand{\mydriver}{pdftex}
\fi

\usepackage[breaklinks,\mydriver]{hyperref}
\hypersetup{colorlinks=true,
            citebordercolor={.6 .6 .6},linkbordercolor={.6 .6 .6},%
citecolor=Darkblue,urlcolor=black,linkcolor=Purple}
\newcommand{\lref}[2][]{\hyperref[#2]{#1~\ref*{#2}}}

\newcommand{\qedsymbol}{$\Halmos$}
\newcommand{\eop}{$\!$\hfill\qedsymbol}

\DeclarePairedDelimiterX{\infdivx}[2]{(}{)}{%
  #1\;\delimsize\|\;#2%
}

\newcommand{\bfpm}{{\textsf{BFPM}}\xspace}
\newcommand{\fppe}{{\textsf{FPPE}}\xspace}
\newcommand{\sppe}{{\textsf{SPPE}}\xspace}
\newcommand{\erce}{{\textsf{ERCE}}\xspace}
\newcommand{\eg}{{\textsf{EG}}\xspace}
\newcommand{\kkt}{{\textsf{KKT}}\xspace}

\newcommand{\fppet}{(\alpha, x)}
\newcommand{\newfppet}{(\alpha', x')}

\newcommand{\mcp}{{\textsf{CP}}\xspace}

\begin{document}


\RUNAUTHOR{Conitzer et al.}

\RUNTITLE{Pacing Equilibrium in First Price Auction Markets}

\TITLE{Pacing Equilibrium in First Price Auction Markets}

\ARTICLEAUTHORS{%
\AUTHOR{Vincent Conitzer}
\AFF{Econorithms LLC, and
Computer Science Department, Duke University
\EMAIL{conitzer@cs.duke.edu}} 
\AUTHOR{Christian Kroer\footnote{This work was done while the author was full time at Facebook Core Data Science.}}
\AFF{Facebook Core Data Science, and
IEOR Department, Columbia University
\EMAIL{christian.kroer@columbia.edu}}
\AUTHOR{Debmalya Panigrahi}
\AFF{Computer Science Department, Duke University
\EMAIL{debmalya@cs.duke.edu}}
\AUTHOR{Okke Schrijvers, Nicolas E. Stier-Moses, Eric Sodomka}
\AFF{Facebook Core Data Science
\EMAIL{okke@fb.com}
\EMAIL{nicostier@yahoo.com}
\EMAIL{sodomka@fb.com}
}
\AUTHOR{Christopher A. Wilkens\footnotemark[0]}
\AFF{Tremor Technologies
\EMAIL{c.a.wilkens@gmail.com}}
} 

\ABSTRACT{%
Mature internet advertising platforms offer high-level campaign management tools to help advertisers run their campaigns, often abstracting away the intricacies of how each ad is placed and focusing on aggregate metrics of interest to advertisers. On such platforms, advertisers often participate in auctions through a proxy bidder, so the standard incentive analyses that are common in the literature do not apply directly. 
In this paper, we take the perspective of a budget management system that surfaces aggregated incentives---instead of individual auctions---and compare first and second price auctions. We show that theory offers surprising endorsement for using a first price auction to sell individual impressions.  In particular, first price auctions guarantee uniqueness of the steady-state equilibrium of the budget management system, monotonicity, and other desirable properties, as well as efficient computation through the solution to the well-studied Eisenberg-Gale convex program. Contrary to what one can expect from first price auctions, we show that incentives issues are not a barrier that undermines the system. Using realistic instances generated from data collected at real-world auction platforms, we show that bidders have small regret with respect to their optimal ex-post strategy, and they do not have a big incentive to misreport when they can influence equilibria directly by giving inputs strategically. Finally, budget-constrained bidders, who have significant prevalence in real-world platforms, tend to have smaller regrets. Our computations indicate that bidder budgets, pacing multipliers and regrets all have a positive association in statistical terms.
}%


\KEYWORDS{mechanism design, auction systems, pacing equilibria}
\ORMSCLASS{Games/group decisions: Bidding/auctions; Marketing: Advertising/promotion; Programming: Nonlinear: Convex}

\maketitle

%


\section{Introduction}
\label{sec:introduction}

The early days of internet advertising were the Wild West. Savvy search advertisers at that time would carefully tune their keywords and bids, building complex tools that precipitated cyclic bidding patterns \citep{edelman2007cycling}. For example, it would pay off to differentiate between users who search for ``flower'' and ``flowers'' in search advertising. Deft display advertisers would exploit real-time bidding (RTB) capabilities---bidding programmatically on a per impression basis---to identify high-value users where competitors lacked data, reaching them at bargain basement prices. Crafty clickers would sabotage competitors' campaigns to gain advantage. Such complex strategic behavior was enabled by low-level bidding tools and driven by advertisers' focus on the value of each customer.

As platforms matured and developed higher-level functionality, programmatic advertising attracted advertisers who cared more about the population of users they reached than about the value of an individual user. Early programmatic systems like search advertising and RTB bidding, where advertisers more-or-less bid on each opportunity, catered to advertisers targeting events deep in the sales funnel where the value of each event was substantial and measurable. Even as these systems came to dominate search and remnant display advertising\footnote{Remnant inventory in display comprised ad opportunities that were not otherwise sold by a salesperson through a guaranteed contract.}
in the first decade of the 21st century, advertisers simply looking to reach an audience would often go through salespeople to buy so-called guaranteed placements. With only a rough measure of the value in expectation of reaching each individual, competitive bidding on a per-user basis made little sense for these advertisers. The gap between these worlds was bridged by platforms that abstracted away the auctions to focus on reach.
Advertisers would input basic data to set up a campaign, and platforms help them run the campaign relying on a proxy bidder.
For example, an advertiser at a social media platform today might simply say ``I want to show my ad to as many junior tech professionals in Colorado as possible for \$1,000'' but not indicate the value per individual.

Facebook's current advertiser workflow exemplifies this mindset as illustrated in Figure~\ref{fig:fb_ad}. The campaign creation flow asks an advertiser four questions: 
(1) what do you want people to do (see, watch, click, buy, etc.),
(2) who do you want to reach (targeting criteria), 
(3) where do you want to show your ad (e.g., Facebook, Instagram, Messenger), and
(4) how much money do you want to spend (budget).
Notably missing in the basic flow is the willingness to pay to reach an individual person (bid).
\footnote{Specifying the bid is available as an advanced feature offered for certain types of campaigns. We refer the reader to Facebook ad product documentation, which includes a list of campaign types in a table where it can be seen in the second column that most of the campaign types are {\em auto-bid}. See \url{https://www.facebook.com/business/m/one-sheeters/facebook-bid-strategy-guide}\,.}

\begin{figure}[t]
\centering\includegraphics[width=5in]{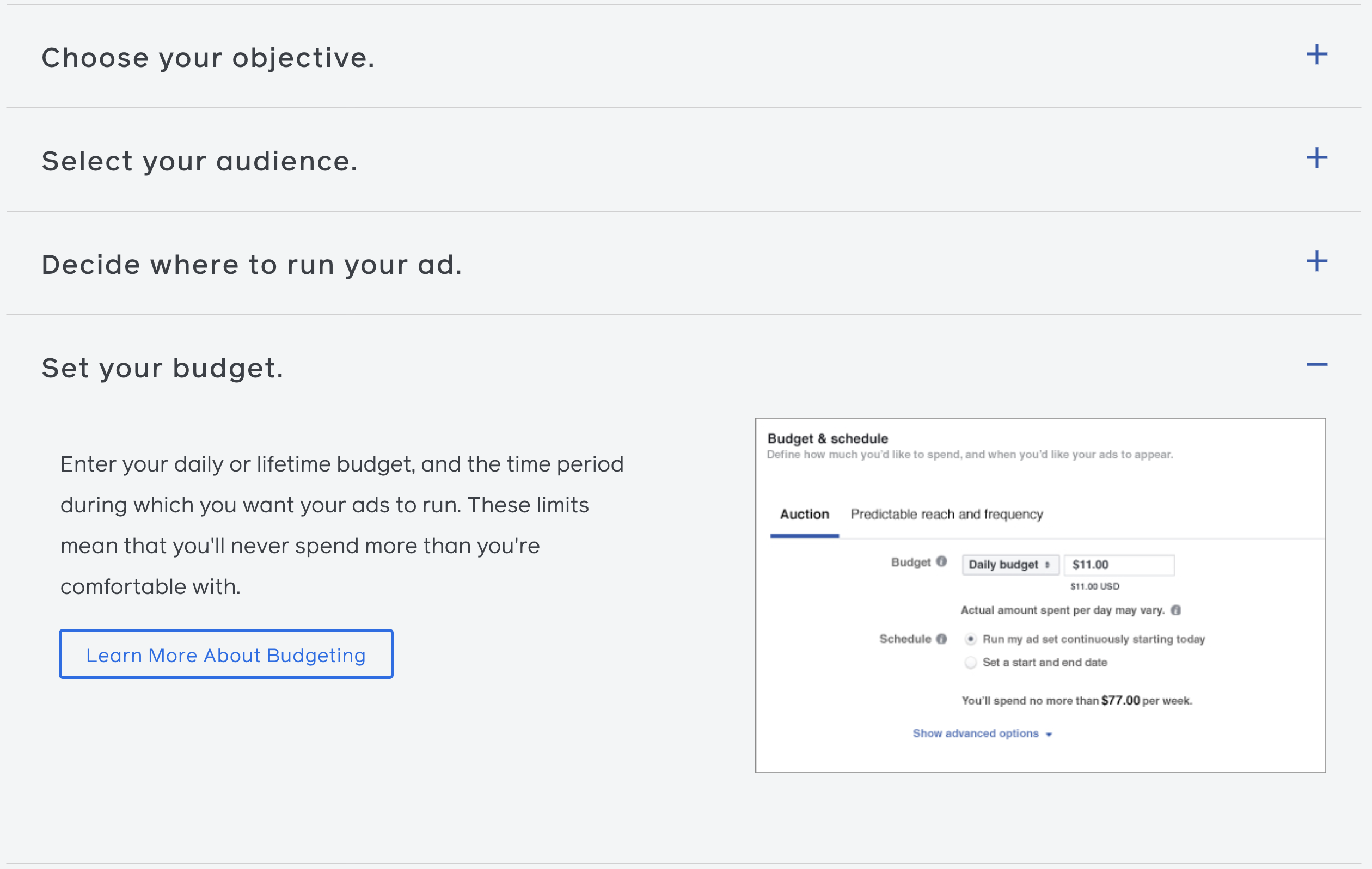}
\caption{Flow to create a Facebook ad as described at \url{https://www.facebook.com/business/ads}\,.}\label{fig:fb_ad}
\end{figure}

This shift from the primacy of the individual user to the primacy of the population being reached substantially changes the mechanism designer’s problem. In the standard theory, the bid is assumed to be the primary strategic implement and budgets, when present, are often assumed to be public knowledge to make analysis tractable. However, when focusing on the population reached, advertisers’ primary strategic levers are budgets and audience targeting, with control of the bid ceded to a proxy bidder. This turns much of the standard auction literature on its head. The advertiser’s bid, which was already the product of a constant per-event bid and a platform-generated impression-specific event probability, is now entirely managed by a proxy bidder tasked with managing the advertiser’s budget across impressions. 
This leads us to the motivation of the present paper: in a system where advertisers' primary or only strategic lever is the budget, what is the best way to allocate and price ads? 

The goal of our study is to understand the strengths and weaknesses of different auction rules in the context of proxy bidders that help run campaigns.
We study a specific question on this topic: {\em when the proxy bidder is designed to shade bids so the advertiser's budget is exhausted at the end of the budget horizon, would the system as a whole perform better if individual impressions are sold through a first price or a second price mechanism?} Since the proxy bidders shade bids on behalf of advertisers to maximize their utilities over the course of the campaign, the standard analyses of first and second price auctions no longer apply. Instead, the focus is the study of equilibria resulting from the independent and strategic behavior of all proxy bidders.
First, we show that when a first price auction is used to sell each impression, the resulting {\em first price pacing equilibrium} (\fppe) has many theoretical guarantees that would not hold in an analogous {\em second price pacing equilibrium} (\sppe), including uniqueness and monotonicity. Second, we show that \fppe are not an arbitrary construction: they correspond to market equilibria, which can be efficiently computed using the well-studied Eisenberg-Gale (\eg) convex program (unlike \sppe which are PPAD-complete to compute). Finally, we run simulations on real data collected from Facebook and Instagram advertising auctions, in order to study the impact on incentives for advertisers.
In these simulations, we find that the ex-post incentive for an individual advertiser to shade their bids via reporting a lower value per conversion or budget is very small in almost all cases.
In an individual first price auction, the winner has a significant incentive to bid strategically when the gap between the first and second bid is large, and closing the gap requires a ``thick’’ market where many bidders are present. However, in a pacing equilibrium, the proxy bidder is picking a single bid to be used in all auctions, so it is enough that the gap between the first and second bids is small in a fraction of the auctions in which the proxy bidder will participate.
We also find that the incentive for misreporting the value per conversion or budget as inputs to the mechanism is vanishingly small. We hypothesize that this conclusion has to do with the coarseness at which manipulations can be performed when buyers do not have the ability to shade bids in individual auctions.

\subsection{Pacing Equilibria in Internet Advertising}

The primary function of proxy bidders in internet advertising is to optimize the way advertisers’ budgets are spent. Most modern implementations are based on control systems that tune one or more control parameters, and apply them to the incoming opportunities. A common budget management method known as bidder selection or {\em throttling} tries to enforce budget constraints by adaptively selecting which advertisers participate in each auction.
This is commonly done by tuning a probability for each ad to randomize the participation in the auctions it targets. The platform adjusts the ad's participation probability continuously: while an ad keeps spending its budget too fast, its participation probability is continuously reduced, and vice-versa. Other platforms instead modify the bid of each ad by applying a shading factor, referred to as a {\em (multiplicative) pacing multiplier}, to manage the budget. Tuning the pacing multiplier changes the spending rate, thereby attempting to exhaust the ad's budget precisely at the end of the budget horizon. (For further details, we expand the description of pacing in the next section.) Between these two options, shading-based systems tend to increase advertisers' utilities by allocating the cheapest impressions to them, while throttling offers advertisers a more representative sample of the available opportunities. There is no one-size-fits-all solution and the design space of budget management systems is large and rife with trade-offs. 

Both methods have precedent. Throttling has been studied in both first and second price settings.
In the context of first price auctions, the celebrated work of \citet{mehta2007adwords} gives an algorithm for online allocation of impressions to bidders to maximize overall ad revenue. Their allocation algorithm, and those in the large body of follow-up work (see the survey in \citealt{mehta2013survey}), can be interpreted as running a first price auction for each impression after removing a subset of bidders from the auction based on their remaining budgets. Bidder selection has also been explored in the context of Generalized Second Price (GSP) auctions, particularly for multi-objective optimization in search engines~\citep{abrams2008comprehensive,azar2009gsp,goel2010gsp,karande2013pacing}. An important feature of bidder selection is that a bidder who is chosen to participate in the auction does so with her original bid, i.e., the platform does not modify bid values of participating bidders.

Our work focuses on bid modification, wherein the platform shades an advertiser's bids in order to preserve her budget for the future. 
This is commonly implemented by scaling an advertiser's bid by a pacing multiplier between $0$ and $1$. Many advertising platforms provide a free option to advertisers to automatically have their bids scaled, and an impressive body of work has focused on advertiser strategies for bid modification in order to maximize their ROI~\citep{rusmevichientong2006,feldman2007budget,hosanagar2008}. In addition to advertiser strategies, \cite{cray2007} also studies the limit point if all advertisers participate in a repeated position auction. In that setting, prices converge to Vickrey-Clarke-Groves (VCG) prices, but there is no heterogeneity in impression opportunities. Closest to our work is that of \cite{borgs2007dynamics} who study first price auctions with budget constraints in a perturbed model. They show that in the limit of their perturbations, prices converge to those of an (equal-rates) competitive equilibrium. The limit point they describe is an \fppe, hence this shows guaranteed existence and the relation to competitive equilibria. It leaves open the question of uniqueness, revenue, and other properties of \fppe. They also point out a convex set of constraints that can be used to compute this limit point, and thus an \fppe. 

Recent work has studied pacing equilibria in the context of second price auctions: \citet{balseiro2015repeated} investigate budget-management in auctions through a fluid mean-field approximation, which leads to elegant existence results and closed-form descriptions of equilibria in certain settings. \citet{balseiro2017budget} studies several budget smoothing methods including multiplicative pacing in a stochastic context where ties do not occur; \citet{balseiro2017dynamic} studies how an individual bidder might adapt their pacing multiplier over time in a stochastic continuous setting; \citet{conitzer2021pacing} studies second price pacing when bidders, goods, budgets, and valuations are known, and proves that equilibria exist under fractional allocations.
In this paper, we study pacing in the context of first price auctions. While
second price auctions displaced first price auctions in Internet advertising
because of their many desirable robustness guarantees, particularly related to
stability~\citep{edelman2007cycling} and strategyproofness, first price auctions
are regaining popularity because they are simple to operate and offer a higher degree of transparency and ease of use than second price auctions~\citep{digiday2018firstprice,adexchanger2017firstprice}. 

While second price auctions displaced first price auctions in Internet advertising
because of their many desirable robustness guarantees, particularly related to
stability~\citep{edelman2007cycling} and strategyproofness, first price auctions
are regaining popularity. In 2017, exchanges such as AppNexus, Index Exchange and OpenX began moving their inventory to from second- to first-price auctions~\citep{adexchanger2017firstprice}. Google transitioned all display advertising publisher inventory to first price auctions in 2019, citing a desire for greater transparency and simplicity in the ad market; they reported a neutral to positive impact on revenue for publishers and a more competitive market~\citep{google2019firstprice}. Mopub (Twitter's ad platform) transitioned all publisher inventory to a first price auction from June to October of 2020, also citing a desire for increased transparency, reduced complexity, and increased fairness~\citep{mopub2020firstprice}. 
To explain this recent shift to first price auctions, \citet{paes2020competitive} analyzed the equilibria of a game played between auction markets, where each market can choose its auction rules. We focus on a game between proxy bidders, as opposed to a game between auction markets. 
We remark that this shift to first price auctions has happened in display advertising, which typically uses real-time bidding, as opposed to a centralized pacing system (although intermediaries likely have their own budget pacing systems for their advertisers). By contrast, in this work, we explore the properties that arise when first price auctions are combined with centralized pacing systems.

In the context of 
position auctions, \citet{dutting2018expressiveness} further showed that first price does not
suffer from the same equilibrium selection problems as GSP and VCG. Moreover, as we will see later, first price auctions provide a clean characterization of equilibrium solutions in the context of pacing, unlike in the case of second price auctions~\citep{conitzer2021pacing}. Indeed, as stated earlier, one of the most developed line of research in ad auctions relates to bidder selection in first price auctions~\citep{mehta2013survey}. Our work complements this line of work by focusing on bid modification instead of bidder selection as the preferred method for budget management.\footnote{In \citet{mehta2007adwords} and most other related papers, bidder selection is performed by using a ``scaling parameter'' for bids based on remaining budgets, but once the winner of an auction is selected, she pays her entire (unscaled) bid for the impression. In contrast, in our work, the individual bids are scaled using the respective pacing multipliers, and the winner pays the scaled bid instead of her original bid for the impression. Another difference is that in the budgeted allocation literature the scaling multiplier changes from one impression to the next since it is only a tool for winner selection, whereas we choose a single pacing multiplier for a bidder that is used to scale her bids for all the impressions in the problem instance.}

\subsection{Contributions}

We introduce the first-price pacing equilibrium (\fppe) problem, where each buyer participates in many first-price auctions, with budget constraints that span the auctions.
The platform seeks to find a vector of \emph{pacing multipliers}, one for each buyer, and buyers bid their value times their pacing multiplier. At a high level, the platform seeks to find multipliers such that buyers are only just paced enough that they do not overspend their budget.
We allow goods to be allocated fractionally if there are ties; this is well-motivated in the ad-auction setting, where a good may be interpreted as representing thousands of impressions.
We show that \fppe can be understood as a maximal point within the larger class of \emph{budget-feasible pacing multipliers} (\bfpm), which is the set of vectors of pacing multipliers that satisfy all budgets, but potentially pace beyond what is necessary in order to satisfy budgets. 

{\em Existence and Uniqueness.}
From \cite{borgs2007dynamics} it follows that \fppe are guaranteed to exist. Our first result shows that not only does an \fppe always exist, but also that it is {\em essentially unique}.\footnote{More precisely, pacing multipliers are unique but there may be different equivalent allocations due to tie breaking. Tie breaking does not impact revenue, social welfare or individual utilities.} In fact, we show that the \fppe exactly coincides with the (unique) {\em maximal} set of pacing multipliers. This also leads to the observation that the \fppe is revenue-maximizing among all \bfpm. Our structural characterization of \fppe in terms of \bfpm is a powerful tool for reasoning about properties of \fppe and is the basis of several of our results and additional ones by \citet{peysakhovich2019fair} in follow-up work on competitive equilibria from equal incomes. In contrast, \sppe are guaranteed to exist but are not unique in general~\citep{conitzer2021pacing}. 
Furthermore, we also show that an \fppe yields a {\em competitive equilibrium}, i.e., every bidder is exactly allocated her {\em demand set}.
\sppe also corresponds to a competitive equilibrium, but only for buyers that are \emph{supply-aware} when computing their demands, whereas \fppe is a competitive equilibrium both for supply-aware and supply-unaware buyers.

\noindent{\em Computability.}
We show an interesting connection between \fppe and a generalization of the classical \eg convex program for quasi-linear utilities~\citep{cole2017convex}. Using Fenchel duality, we use the \eg program to infer that the unique maximal \bfpm, which is also revenue maximizing, yields the unique \fppe. Moreover, this connection with the \eg program immediately yields a (weakly) polynomial algorithm to compute the \fppe. This also contrasts with \sppe, for which maximizing revenue (or other objectives like social welfare) is known to be NP-hard~\citep{conitzer2021pacing}, and very recent results show that the general problem is PPAD-complete~\citep{chen2021complexity}.
While \citet{borgs2007dynamics} already showed a convex program, this connection to \eg is important because it lends itself to scalable first-order methods~\citep{nesterov2018computation,kroer2021computing}.

\noindent{\em Monotonicity and Sensitivity:} We show that \fppe satisfies many notions of monotonicity, both in terms of revenue and social welfare. Adding an additional good weakly increases both revenue and social welfare, and adding a bidder or increasing a bidder's budget weakly increases revenue.
Again, this stands in sharp contrast with \sppe, which generally do not satisfy such monotonicity conditions.

In fact, in an \fppe not only is the revenue monotonically non-decreasing with budgets, but it also changes smoothly in the sense that increasing the budget by a constant $\Delta$ can only increase the revenue by $\Delta$. This does not hold for \sppe where the revenue can increase by a substantial amount even for a small change in the budgets.

\noindent{\em Shill-proofness}: Using monotonicity, we also establish that there is no incentive for the platform to enter fake bids to an \fppe mechanism.

\noindent{\em Simulations}: To test the properties mentioned earlier and compare equilibria computationally, we rely on the \eg convex program formulation to perform simulations on
realistic instances constructed from elements of real-world data. We find that the ex-ante and ex-post regret associated with \fppe is very small when bidders can only change the scale
of their utilities, as is usually the case in ad markets when bidding on clicks or conversions.
Finally, we compare the revenue and social welfare of \fppe and \sppe across instances. 
We find that \fppe always provide higher revenue, while welfare splits about evenly on which solution concept performs better.
The contributions and the comparison between \fppe and \sppe are summarized in Table~\ref{tab:comparison}.

In Section~\ref{sec:assumptions discussion} we discuss the assumptions we have made, and put them in perspective versus more general or alternative assumptions.

\begin{table}
\begin{center}
\begin{tabular}{|c|c|c|}
\hline
  & \sppe & \fppe \\ 
  \hline
  \hline
 Exists? & Yes & Yes  \\  
 \hline
 Buyers best responding? & Yes & No \\
 \hline
 Is market eq.? & Yes & Yes (even for supply-unaware buyers)\\
 \hline
 Is unique? & No & Yes, in utilities, multipliers, and prices\\
 \hline
 Is efficiently computable? & PPAD-complete & Convex program \\
 \hline
 Is welfare monotone? & No & Yes, in goods \\
 \hline
 Is revenue monotone? & No & Yes, in goods/bidders/budgets   \\
 \hline
 Is shill proof? & No & Yes \\
 \hline
 \hline
 Simulated regret/IC & No regret, very small IC & small regret, small IC \\
 \hline
 Simulated revenue & \multicolumn{2}{|c|}{\sppe $\leq$ \fppe} \\
 \hline
Simulated welfare & \multicolumn{2}{|c|}{Ambiguous}  \\
 \hline
\end{tabular}
\caption{A comparison of \fppe and \sppe. 
In the second row, we note that under \sppe, each bidder's multiplier is a mutual best response given other-agent bids are fixed---and would be even if the set of deviating strategies were expanded to allow the bidder to submit arbitrary per-auction bids (for a proof, see Proposition 1 in \citealt{conitzer2021pacing}). Under \fppe, a buyer may increase its utility by shading its multiplier (or shading individual bids, if it could submit arbitrary bids for each auction).
}\label{tab:comparison}
\end{center}
\end{table}

\section{First price Pacing Equilibria}\label{s:model}

We consider a single-slot auction market in which a set of bidders
$N=\{1,\ldots,n\}$ target a set of (divisible) goods $M=\{1,\ldots,m\}$. Each bidder $i$ has
a valuation $v_{ij}\ge 0$ for each good $j$, and a budget $B_i>0$ to be spent
across all goods. We assume that the goods are sold through independent (single
slot) first price auctions, and the valuations and budgets are assumed to be
known to the auctioneer. When multiple bids are tied for an item, we assume that the item 
can be fractionally allocated, and we allow the auctioneer to choose the fractional allocation 
(although our results on equivalence to competitive equilibrium show that the  fractional choices made by the auctioneer are optimal for the bidders as well).

The goal is to compute a vector of \emph{pacing multipliers} $\alpha$ that smooths out
the spending of each bidder so that they stay within budget. A pacing
multiplier for a bidder $i$ is a real number $\alpha_i\in [0,1]$ that is used to
scale down the bids across all auctions: for any $i,j$, bidder $i$ participates
in the auction for good $j$ with a bid equal to $\alpha_i v_{ij}$; we refer to
these bids as {\em multiplicatively paced}. We define feasibility as follows:

\begin{definition}\label{def:bfpm}
	A set of {\em budget-feasible first price pacing multipliers} (\bfpm) is a tuple $(\alpha, x)$, of pacing multipliers $\alpha_i \in [0,1]$ for each bidder $i \in N$, and fractional allocations $x_{ij}\in [0,1]$ for each bidder $i \in N$ and good $j \in M$, that satisfies the following properties:
	\begin{itemize}
		\item (Prices) Unit price $p_j = \max_{i \in N} \alpha_i v_{ij}$ for each good $j\in M$.
		\item (Goods go to highest bidders) If $x_{ij}>0$, then $\alpha_i v_{ij} = \max_{i'\in N} \alpha_{i'} v_{i'j}$ for each bidder $i \in N$ and good $j \in M$.
		\item (Budget-feasible) $\sum_{j\in M} x_{ij} p_j \le B_i$ for each bidder $i\in N$.
		\item (Demanded goods sold completely) If $p_j > 0$, then $\sum_{i\in N} x_{ij} = 1$ for each good $j \in M$.
		\item (No overselling) $\sum_{i\in N} x_{ij} \leq 1$ for each good $j\in M$.
	\end{itemize}
\end{definition}

Within the feasible space, we are particularly interested in outcomes that are stable in some sense. Specifically, we are interested in solutions where no bidder is unnecessarily paced, which we call first price pacing equilibria (\fppe). This captures that the platform does not want to pace an advertiser who is not spending the full budget.

\begin{definition}\label{def:fppe}
  A {\em first price pacing equilibrium} (\fppe) is a \bfpm tuple $(\alpha, x)$, of pacing multipliers $\alpha_i$ for each bidder~$i$, and fractional allocation $x_{ij}$ for bidder $i$ and good $j$ with this additional property:
  \begin{itemize}
  	 \item (No unnecessary pacing) If $\sum_{j\in M} x_{ij}p_j < B_i$, then $\alpha_i = 1$ for each bidder $i\in N$.
   \end{itemize}
\end{definition}

\subsection{Existence, uniqueness, and structure of \fppe}
\cite{borgs2007dynamics} show that \fppe are guaranteed to exist. We show that not only are \fppe guaranteed
to exist, but that they are also unique and maximize the seller's revenue over all \bfpm. In the following, inequalities should be interpreted \emph{component-wise} when applied to vectors or sets.

\begin{lemma}
	There exists a Pareto-dominant \bfpm $\fppet$ (i.e., $\alpha \geq \alpha'$ for any \bfpm $\newfppet$\,).
\end{lemma}
\proof{Proof.}
	First, we will show that given any two \bfpm $(\alpha^{(1)}, x^{(1)})$ and $(\alpha^{(2)}, x^{(2)})$, there exists a \bfpm with pacing multipliers  $\alpha^* = \max(\alpha^{(1)}, \alpha^{(2)})$ that are the component-wise maximum of $\alpha^{(1)}$ and $\alpha^{(2)}$. 
    Note that, on each item, the resulting paced bid for a bidder is the higher of her two paced bids on this item in the original two \bfpm{s}.
    Let the corresponding allocation $x^*$ be: for each good $j$, identify which of the two \bfpm{s} had the highest paced bid for $j$, breaking ties towards the first \bfpm; then, allocate the good to the same bidder as in that \bfpm (if the good was split between multiple bidders, allocate it in the same proportions), at the same price.  Note that these prices coincide with the winning bidders' paced bids in the new solution.  Thus, we charge the correct prices, goods go to the highest bidders, demanded goods are sold completely, and there is no overselling.
     
All that remains to be verified is that the new solution is budget-feasible. 
Consider bidder $i$ in the bidder-wise-max $\alpha^*$.  In either $\alpha^{(1)}$ or $\alpha^{(2)}$, bidder $i$ had the exact same multiplier---say it was in $\alpha^{(b)}$ for $b \in \{1,2\}$ (breaking ties towards $1$).
To show budget feasibility, we now prove the following: if a bidder $i$ wins fraction $x$ of an item $j$ in $(\alpha^*, x^*)$, 
then she must also win at least fraction $x$ of item $j$ in $(\alpha^{(b)}, x^{(b)})$. 
The key observation is that for any bidder and item, her
paced bid in $(\alpha^*, x^*)$ is at least as much as her paced bid in either $(\alpha^{(1)}, x^{(1)})$ 
or $(\alpha^{(2)}, x^{(2)})$---we call this the {\em monotone} property.
Now, consider two situations. First, 
if $i$ did not have the highest bid on an item $j$ in $(\alpha^{(b)}, x^{(b)})$, then the bidder who had a higher bid in 
$(\alpha^{(b)}, x^{(b)})$
continues to have a higher bid in $(\alpha^*, x^*)$ by the monotone property. Hence, $i$ does not win this item under 
$(\alpha^*, x^*)$.
Second, suppose $i$ did have the highest bid on an item $j$ in $(\alpha^{(b)}, x^{(b)})$, and won a fraction $x$ of the item. 
In this case, every other bidder who matched 
$i$'s bid in $(\alpha^{(b)}, x^{(b)})$ continues to have at least as high a bid in $(\alpha^*, x^*)$ by the monotone property. Therefore,
$i$ shares item $j$ in $(\alpha^*, x^*)$ with at least the set of bidders she shared it with in $(\alpha^{(b)}, x^{(b)})$. There are three further
subcases. First, if the highest paced bid for $j$ is unique, then the highest must be achieved by $(\alpha^{(b)}, x^{(b)})$; in this case, $i$ gets 
exactly fraction $x$ of item $j$. Second, if the highest paced bid for $j$ is tied and $b=1$, then item $j$ is divided 
exactly as in $(\alpha^{(1)}, x^{(1)})$. Again, $i$ gets exactly fraction $x$ of item $j$. Third, if the highest paced bid for $j$ is tied
but $b=2$, then we claim that $i$ does not get item $j$ in $(\alpha^*, x^*)$. This is because the allocation of $j$ under $(\alpha^*, x^*)$
is identical to that under $(\alpha^{(1)}, x^{(1)})$, but since $b =2$, we have $\alpha^{(1)}_i < \alpha^{(2)}_i$ which implies that 
$i$ does not have the highest bid for $j$ in $(\alpha^{(1)}, x^{(1)})$. In summary, in any of these three cases, the fraction of 
item $j$ that bidder $i$ wins under $(\alpha^*, x^*)$ is at most the fraction she wins under $(\alpha^{(b)}, x^{(b)})$.
As a consequence, $i$ spends no more under the new \bfpm than under \bfpm $b$, which is budget-feasible.  Hence, the new \bfpm is budget-feasible.

We now complete the proof. Let $\alpha_i^* = \sup\{\alpha_i \  |  \ \alpha \text{ is part of a \bfpm}\}$.  We will show that $\alpha^*$ is part of a \bfpm, proving the result.  For any $\epsilon>0$ and any $i$, there exists a \bfpm where $\alpha_i> \alpha_i^* - \epsilon$.  By repeatedly taking the component-wise maximum for different pairs of $i$, we conclude there is a single \bfpm $(\alpha^\epsilon, x^\epsilon)$ such that for every $i$, $\alpha_i^\epsilon > \alpha_i^* - \epsilon$.  Because the space of combinations of multipliers and allocations is compact, the sequence $(\alpha^\epsilon, x^\epsilon)$ (as $\epsilon \rightarrow 0$) has a limit point $(\alpha^*, x^*)$.  This limit point satisfies all the properties of a \bfpm by continuity.
\eop\endproof

In addition to there being a maximal set of pacing multipliers over all \bfpm, this maximal \bfpm is actually an \fppe.

\begin{lemma}[Guaranteed existence of \fppe]\label{lem:existence}
The Pareto-dominant \bfpm $\fppet$ with maximal pacing multipliers $\alpha$ has no unnecessarily paced bidders, so it forms an \fppe.
\end{lemma}
\proof{Proof.}
Suppose bidder $i$ is unnecessarily paced under the maximal \bfpm $\fppet$. If $i$ is not tied for any good, then we can increase her pacing multiplier by a sufficiently small $\epsilon > 0$ such that $i$ is still not tied for any item and is within budget, contradicting the fact that $\alpha$ was maximal. So bidder $i$ must be tied for at least one good. 
Define $N(i)$ as all the bidders that are tied for any good with bidder $i$, i.e., $N(i) = \{\text{bidder }k\ :\ \exists {\text{ good } j}\text{ with }\alpha_i v_{ij} = \alpha_{k} v_{kj}\}$. Now take the transitive closure $T$ of this set $N(i)$, i.e., include all bidders who are tied for an item with a bidder in $N(i)$, etc. Next, redistribute the items that are tied such that none of the bidders in $T$ is budget constrained, while still allocating items completely.
This is always possible since we can slightly increase the share of $i$ for all items she is tied on, while simultaneously reducing the share of all other bidders tied with her. In the next step, we can slightly increase the share of all these bidder for other items they are tied on while reducing the share of the new bidders they are tied with for those items, and so on.
Next, increase the pacing multipliers of all bidders in $T$ by a small enough $\delta > 0$ so that all bidders in $T$ are still not budget-constrained, and no new ties are created; call this set of pacing multipliers $\alpha'$ and the redistribution of goods $x'$. This contradicts that $\alpha$ was the maximal \bfpm to begin with, as $(\alpha', x')$ is a \bfpm yet it has pacing multipliers that are higher than in $(\alpha, x)$.
\eop\endproof

The converse of Lemma~\ref{lem:existence} is also true: any \bfpm for which at least one bidder has a pacing multiplier $\alpha_i$ lower than the maximal \bfpm must have an unnecessarily paced bidder.

\begin{lemma}
Consider two \bfpm $(\alpha^{(1)}, x^{(1)})$ and $(\alpha^{(2)}, x^{(2)})$, where $\alpha^{(1)} \geq \alpha^{(2)}$ and $\alpha^{(1)}_i > \alpha^{(2)}_i$ for some bidder $i$. Then, $(\alpha^{(2)}, x^{(2)})$ must have an unnecessarily paced bidder.
\end{lemma}
\proof{Proof.}
Consider the set $I$ of all bidders whose pacing multipliers are strictly lower in $\alpha^{(2)}$ than in $\alpha^{(1)}$ (by definition there must be at least one bidder in this set). Collectively, $I$ wins fewer (or the same) items under $\alpha^{(2)}$ than under $\alpha^{(1)}$ (the bids from outside $I$ have stayed the same, those from $I$ have gone down), and at lower prices.  Since $(\alpha^{(1)}, x^{(1)})$ was budget feasible, $I$ was not breaking its collective budget before. Since $I$ is spending strictly less, at least some bidder must not spend their entire budget and thus is unnecessarily paced.
\eop\endproof

This implies that the pacing multipliers of \fppe are uniquely determined.

\begin{corollary}[Essential Uniqueness]\label{cor:uniqueness}
	The pacing multipliers of any \fppe are uniquely determined and correspond to the pacing multipliers of the maximal \bfpm.
\end{corollary}

While the pacing multipliers are uniquely determined, the allocation is not:
Tie-breaking may give different goods to different bidders. However,
tie-breaking is inconsequential in the sense that the bidder utilities (and thus
social welfare), item prices (and thus revenue), and the set of
budget constrained bidders, are all uniquely determined.

Given two \bfpm, if the pacing multipliers of one dominate the other, then the revenue of that \bfpm must also be at least as high. In the following, let $\text{Rev}(\alpha, x)$ refer to the revenue of a \bfpm $\fppet$.

\begin{lemma}\label{lem:revenue}
	Given two \bfpm $(\alpha^{(1)}, x^{(1)})$ and $(\alpha^{(2)}, x^{(2)})$, where $\alpha^{(1)} \geq \alpha^{(2)}$, we must have that $\text{Rev}(\alpha^{(1)}, x^{(1)}) \ge \text{Rev}(\alpha^{(2)}, x^{(2)})$.
\end{lemma}
\proof{Proof.}
	Since $\alpha^{(1)} \geq \alpha^{(2)}$, prices under $(\alpha^{(1)}, x^{(1)})$ must be at least as large as under $(\alpha^{(2)}, x^{(2)})$. By the definition of \bfpm, all demanded items must be sold completely. Therefore, under $(\alpha^{(1)}, x^{(1)})$ we sell at least all the items that we sold under $(\alpha^{(2)}, x^{(2)})$ at prices that are at least as high as those under $(\alpha^{(2)}, x^{(2)})$. Hence, $\text{Rev}(\alpha^{(1)}, x^{(1)}) \ge \text{Rev}(\alpha^{(2)}, x^{(2)})$.
\eop\endproof

\begin{corollary}[Revenue-maximizing]\label{cor:revenue-maximizing}
	The \fppe is revenue-maximizing among all \bfpm.
\end{corollary}

The following theorem summarizes the main properties of \fppe that follow from Lemma~\ref{lem:existence} and Corollaries~\ref{cor:uniqueness} and~\ref{cor:revenue-maximizing}.

\begin{theorem}
	Given input $(N, M, V, B)$, an \fppe is guaranteed to exist. In addition, the uniquely-determined maximal pacing multipliers $\alpha$ maximize the revenue over all \bfpm. 
\end{theorem}

\section{Properties of First-Price Pacing Equilibria}
\label{sec:properties}

We first show that an \fppe is also a competitive equilibrium.  In fact, we show that the concept of \fppe is equivalent to a natural refinement of competitive equilibrium.

\begin{definition}
A {\em competitive equilibrium} consists of prices $p_j$ of goods and feasible allocations $x_{ij}$ of goods to bidders such that the following properties hold:
\begin{enumerate}
\item Each bidder maximizes her utility under prevailing prices, that is, for all $i\in N$ it holds that  
$$\mathbf x_i \in \arg \max_{\mathbf x_i \in [0,1]^m}\left\{\sum_{j\in M} (v_{ij}-p_j) x_{ij}: \sum_{j\in M} p_j x_{ij} \leq B_i\right\}.$$
\item Every item with a positive price is sold completely, that is, $p_j >0 \Rightarrow \sum_i x_{ij}=1$ for all $j\in M$.
\end{enumerate}
\end{definition}
We now introduce a refinement of competitive equilibrium that requires that {\em each individual dollar} of a bidder is spent (or not spent) in a way that maximizes the utility obtained by the bidder for that dollar.  Thus, there exists a rate $\beta_i$ for each bidder that indicates her return on a dollar.
\begin{definition}
An {\em equal-rates competitive equilibrium} (\erce)
is a competitive equilibrium such that for every bidder $i\in N$, there is a number $\beta_i$ such that:
\begin{enumerate}
\item If $x_{ij}>0$, then $v_{ij}/p_j = \beta_i$.
\item If $i$ does not spend her entire budget, then $\beta_i=1$.
\end{enumerate}
\end{definition}

Through the previous definition we obtain the following characterization of \fppe.

\begin{theorem}\label{thm:erce iff fppe}
A combination of prices $p_j$ and allocations $x_{ij}$ is an \erce if and only if it is an \fppe.
\end{theorem}
\proof{Proof.}
We first note that budget feasibility, the no-overselling condition, and the condition that items with a positive price must be sold completely, appear in the definitions of both concepts, so we only need to check the other conditions.

We first prove that an \fppe is also an \erce.  
Let $\beta_i=1/\alpha_i$.
First, consider a bidder with $\alpha_i=1$.  If $x_{ij}>0$ for some $j$, then $v_{ij}/p_j = v_{ij}/v_{ij} = 1 = \beta_i$, proving both conditions in the definition of an \erce.  Moreover, for any item $j$, we have $v_{ij}/p_j \leq v_{ij} / v_{ij} = 1$.  Therefore, the bidder is spending optimally given the prices.

Next, consider a bidder with $\alpha_i<1$. 
If $x_{ij}>0$ for some $j$, then $v_{ij}/p_j = v_{ij}/(\alpha_i v_{ij}) = \beta_i$, proving the first condition in the definition of an \erce.  Moreover, by the definition of \fppe, such a bidder must spend her entire budget, proving the second condition.  Moreover, for any item $j$, we have $v_{ij}/p_j \leq v_{ij} / (\alpha_i v_{ij}) = \beta_i$.  Hence, the bidder is spending all her budget on the optimal items for her, and leaving money unspent would be suboptimal because $\beta_i>1$.  Therefore, the bidder is spending optimally given the prices.  We conclude that an \fppe is also an \erce.

We next prove that an \erce is also an \fppe.
For a bidder $i$ with $\sum_j x_{ij}>0$, consider the set of items $S_i = \{j: x_{ij}>0\}$.  By the \erce property, we have that for $j \in S_i$, 
 $v_{ij} / p_j = \beta_i$.  Let $\alpha_i=1/\beta_i$.
  We must have $\alpha_i \leq 1$, because otherwise $i$'s dollars would be better left unspent, contradicting the first property of competitive equilibrium.  Also, if $\alpha_i<1$ then all of $i$'s budget must be spent, establishing that no bidder is unnecessarily paced.
  For a bidder with $\sum_j x_{ij}=0$, define $\alpha_i = \beta_i = 1$.   
 
Now, we show that no part of an item $j$ can be won by a bidder $i$ for whom $\alpha_i v_{ij}$ is not maximal; if it were, by the first \erce condition we would have $v_{ij}/p_j=\beta_i \Leftrightarrow \alpha_i v_{ij} = p_j$ and another bidder $i'$ for whom $\alpha_{i'} v_{i'j} > \alpha_i v_{ij} = p_j$.  Hence, $v_{i'j} / p_j > 1 / \alpha_{i'} = \beta_{i'}$, but this would contradict that $i'$ is receiving an optimal allocation under the prices.

Next, we prove that the prices are set correctly for an \fppe.
For any item $j$ that is sold completely, consider a bidder with $x_{ij}>0$. Again, by the first \erce condition we have $v_{ij}/p_j=\beta_i \Leftrightarrow \alpha_i v_{ij} = p_j$, and we have already established that this bidder must maximize $\alpha_i v_{ij}$.
If the item $j$ is not entirely sold, then by the second condition of competitive equilibrium it must have price $0$.  This in turn implies that all bidders have value $0$ for it, for otherwise there would be a bidder $i$ with $\beta_i = v_{ij} / p_j =\infty$, who hence should be able to obtain 
a utility of $\infty$ since every one of the bidder's dollars must result in that amount of utility for her---but this is clearly impossible with finitely many resources.
Thus, we have established all the conditions of an \fppe. 
\eop\endproof

Next, we show that the platform does not benefit from adding fake bids under an \fppe mechanism, unlike in the case of second-price payments.

\begin{definition}
A solution concept is {\em shill-proof} if the seller does not benefit from adding fake bids.
\end{definition}

\begin{proposition}
\fppe are shill-proof.
\end{proposition}
\proof{Proof.}
Note that if we start from an \fppe and remove both a bidder and the items she wins, we still have an \fppe, since the remaining bidders are spending the same as before, and the remaining items are allocated as before and thus fully allocated. Consider an instance of a market with three \fppe: (a) an \fppe with shill bids, (b) an \fppe without shill bids, and (c) the \fppe generated by removing both the shill bids and the items they won from (a). Notice that the seller makes the same revenue in (a) and (c). Moreover, by Proposition~\ref{prop:rev-good} we know that the revenue of (b) is at least as much as the revenue of (c), and therefore also at least as much as the revenue of (a). Thus, the seller cannot benefit from shill bids.
\eop\endproof

\citet{akbarpour2018credible} observe that a first price single-item auction satisfies a property referred to as {\em credible}. This means that the platform cannot benefit from misrepresenting what other agents have done in the mechanism. The example below illustrates that the \fppe mechanism does not necessarily satisfy this property.  

\begin{example}
Suppose $B_1=2, v_{11}=2, v_{12}=2$ and $v_{22}=1$.  The \fppe sets $p_1=p_2=1$ and allocates both items to bidder $1$.  But the auctioneer could lie to bidder $1$ claiming that someone else had bid $3$ for item $2$, and charge bidder $1$ a price of $2$ for item $1$. Meanwhile, she could charge bidder $2$ a price of $1$ for item $2$, for a higher revenue overall.
\end{example}

However, an \fppe does have a {\em price predictability} guarantee: given any allocation, a bidder either pays its full value or pays her budget. Even though individual item prices may not be known, this guarantees a degree of transparency to bidders about the price they will pay.

An \fppe is also robust to deviations by groups of bidders who might form a coalition to benefit themselves:

\begin{definition}
An allocation with a set of payments is in the core if no group of bidders has an incentive to form a coalition with the seller to attain an outcome that is strictly better for all agents in the coalition.
\end{definition}

\begin{proposition}
An \fppe is in the core.
\end{proposition}
\proof{Proof.}
Since an \fppe is a competitive equilibrium, if we treat money as a good then we have a traditional locally non-satiated Walrasian equilibrium in an exchange economy. Since a locally non-satiated Walrasian Equilibrium is in the core, an \fppe is also in the core. 
(See, e.g., \citet{mascolell-book} for a general exposition of Walrasian equilibria and cooperative game theory, and 
\citet{powell-lecturenotes} for an exposition under assumptions that match those of this paper.)
\eop\endproof

\section{Monotonicity and Sensitivity Analysis}

In the previous section, we showed that \fppe are guaranteed to exist, that they are essentially unique (up to ties that are largely inconsequential), and that they satisfy a number of attractive properties. We now look at how well-behaved \fppe are under changing conditions.
We would ideally like the solution concept to be stable, so that changes in the input do not produce disproportionate changes in the output. We will show that this is largely the case. This is in stark contrast to \sppe, where \citet{conitzer2021pacing} showed that the equilibrium can be very sensitive.
First, \sppe is not unique, and the revenue and welfare can vary drastically
across equilibria. Second, even when there is a unique \sppe, small
changes in the budget can cause disproportionately large changes in revenue.

\subsection{Monotonicity}

We investigate whether \fppe are monotonic when adding bidders or goods, or
when increasing budgets or valuations.
Table~\ref{table:mon-results} summarizes our results.

\paragraph{Revenue}

Revenue monotonicity is maintained when adding bidders, goods, and budget, but not for incremental additions to valuations. Our proofs of revenue monotonicity all rely on Corollary~\ref{cor:revenue-maximizing}: the fact that multipliers in an \fppe are maximal among all \bfpm. Bidder and budget monotonicity both follow from a particularly simple argument: the original solution remains a \bfpm, and thus the maximality of \fppe over \bfpm implies that monotonicity is maintained.

\begin{table}
	\centering
	\begin{tabular}{|l|c|c|c|c|}
		\hline 
		& Add Bidder & Add Good & Incr.\ Budget  & Incr.\ Value $v_{ij}$ \\ 
		\hline \hline
		Revenue & $\ge 0$ & $\ge 0$ & $\ge 0$ & Can go down \\ 
		\hline 
		Social Welfare & Can go down & $\ge 0$ & Can go down &  Can go down\\ 
		\hline 
	\end{tabular} 
	\caption{Overview of monotonicity results.}
	\label{table:mon-results}
\end{table}

\begin{proposition}\label{prop:rev-good}
	In an \fppe, adding a good weakly increases revenue.
\end{proposition}
\proof{Proof.}
Let $(\alpha, x)$ be the \fppe for $N$, $M$, and let $(\alpha', x')$ be the \fppe for $N$, $M \cup \{j\}$ which includes the new good $j\not\in M$. We first prove that $\alpha'_i \le \alpha_i$ for all bidders $i\in N$: Suppose there are bidders whose multipliers go up (strictly); consider the set of all such bidders $S$. Collectively, these bidders are now winning weakly more goods because there are more goods and nobody else's paced bids went up. That means they are, collectively, paying strictly more (they are bidding higher and it is first price). But this is impossible, because all of them were running out of budget before since they were paced.

Using the fact that $\alpha \ge \alpha'$, any bidder who was paced in $\alpha$ is still paced in $\alpha'$ and spending her whole budget. Let $T$ be the set of buyers whose pacing multiplier has not changed, i.e., $T = \{i \in N\ |\ \alpha_i = \alpha'_i\}$. They must win weakly more items: Any item they were tied originally with bidders outside $T$ must now go completely to bidders in $T$. Additionally, bidders in $T$ may win (part of) the new item. Since the pacing multipliers of bidders in $T$ did not change, their prices did not change, hence winning weakly more items means they're spending weakly more.

So bidders whose pacing multiplier changed are spending the same, and the remaining bidder spend weakly more. Hence revenue is weakly higher.
\eop\endproof

\begin{proposition}\label{prop:rev-bidder}
	In an \fppe, adding a bidder weakly increases revenue.
\end{proposition}
\proof{Proof.}
	Let $N$ be the original set of bidders, $i \not\in N$ a new bidder, and $M$ the set of goods. Let $(\alpha, x)$ be the \fppe on $N$ and $M$. After adding bidder $i$, for each bidder  $k \in N \backslash \{i\}$ and good $j \in M$, let $\alpha'_k = \alpha_k$ and $x'_{kj} = x_{kj}$. Set $\alpha'_i = x_{ij} = 0$ for bidder $i$ and good $j \in M$, to obtain $( \alpha', x')$. By construction $(\alpha', x')$ is a \bfpm, so by Lemma~\ref{cor:revenue-maximizing}, the revenue of the \fppe for $N\cup \{i\}$ and $M$ must be at least as high. 
\eop\endproof

\begin{proposition}\label{prop:rev-budget}
	In an \fppe, increasing a bidder's budget from $B_i$ to $B'_i > B_i$ weakly increases revenue.
\end{proposition}
\proof{Proof.}
	Let $(\alpha, x)$ be the \fppe where the budget of bidder $i$ is $B_i$. After increasing the budget to $B'_i$, the solution $(\alpha, x)$ is still a \bfpm. Therefore, by Lemma~\ref{cor:revenue-maximizing}, the revenue of the new \fppe weakly increases.
\eop\endproof

\begin{proposition}\label{prop:rev-value}
	In an \fppe, increasing a bidder $i$'s value for some good $j$ from $v_{ij}$ to $v'_{ij} > v_{ij}$ can decrease revenue.
\end{proposition}
\proof{Proof.}
	Consider the following instance: 2 bidders, 2 goods, $v_{11} = 10, v_{12} = 5$, $v_{21} = 0, v_{22} = 5$, with $B_1 = 10, B_2 = 5$. The \fppe consists of $\alpha_1 = \alpha_2 = 1$, $x_{11} = x_{22} = 1, x_{12} = x_{21} = 0$. Both bidders are spending their whole budget so the total revenue is $15$.
	
	However, consider $v'_{12} = 10 > v_{12}$. The \fppe is now $\alpha_1 = \frac12, \alpha_2 = 1, x_{11} = x_{22} = 1, x_{12} = x_{21} = 0$. The bidders still receive the same goods, but the price for the first good dropped to $5$ for a total revenue of $10$ instead of $15$.	
\eop\endproof

\paragraph{Social Welfare}
For social welfare, monotonicity is only maintained for goods. Adding bidders, or increasing budgets or valuations, can lead to drops in social welfare. The cause of non-monotonicity is that there can be a mismatch between valuation and budget: a high-value but low-budget bidder can be lose out to a low-value high-budget bidder.

\begin{proposition}\label{prop:sw-bidder}
	In an \fppe, adding a bidder can decrease social welfare by a factor of $\frac12$.
\end{proposition}
\proof{Proof.}
	Consider the following instance: 1 bidder, 1 good. We have $v_{11} = K$ for some parameter $K>2$ and $B_1 = 1$. The \fppe is $\alpha_1 = 1/K, x_{11} = 1$ and the social welfare is $K$.
	
	Now add bidder $2$ with $v_{21} = 2, B_2 = 1$. The new \fppe is $a_1 = \frac2K, a_2 = 1$ and $x_{11} = x_{21} = \frac12$. Social welfare now is $\frac{K}2 + \frac12$. As $K\rightarrow \infty$, the new social welfare is half of what it was before. 
\eop\endproof

\begin{proposition}\label{prop:sw-good}
	In an \fppe, adding a good weakly increases social welfare.
\end{proposition}
\proof{Proof.}
	Fix $N$, $M$, and an additional good $j \not\in M$. Let $\fppet$ be the \fppe for $N$ and $M$, $\newfppet$ the \fppe for $N, M\cup\{ j \}$ and let $S$ be the set of bidders who are paced in $\fppet$, i.e., $S = \{ i\ |\ \alpha_i < 1\}$. We will compare the contribution to social welfare of $S$ and $N \backslash S$ separately.

	First we look at the set $S$. From the proof of Proposition~\ref{prop:rev-good}, adding a good weakly decreases pacing multipliers. Since the bidders in $S$ spent their entire budget in $\fppet$, they must also spend their entire budget in $\newfppet$. The bang-per-buck of bidder $i$ is $\frac1{\alpha_i}$, since by Definition~\ref{def:bfpm} they pay $\alpha_i\cdot v_{ij}$ per unit of good $j$, and they receive $v_{ij}$ of value per unit of good $j$. Since pacing multipliers weakly decreased, the bang-per-buck of bidders in $S$ weakly increased, and as they spend their entire budget, their contribution to social welfare weakly increased.
	
	Now, we do it for the set $N \backslash S$. By Proposition~\ref{prop:rev-good}, the total revenue weakly increased. Since the bidders in $S$ spend exactly the same amount as before, the increase in revenue must have come from bidders in $N \backslash S$. Moreover, they were unpaced in $\fppet$ and so had a bang-per-buck of $1$. In $\newfppet$, they have bang-per-buck at least $1$, hence their contribution to social welfare weakly increased.
	
	Since the contribution to social welfare weakly increased for both sets $S$ and $N\backslash S$, the total social welfare weakly increased.
\eop\endproof

\begin{proposition}\label{prop:sw-budget}
	In an \fppe, increasing a bidder's budget from $B_i$ to $B'_i > B_i$ can decrease social welfare.
\end{proposition}
\proof{Proof.}
	Consider the following instance (which is similar to the one in Proposition~\ref{prop:sw-bidder}): 2 bidders, 1 good. We have values $v_{11} = K, v_{21}=2$ and budgets $B_1 = B_2 = 1$. The \fppe is $a_1 = \frac2K, s_2 = 1, x_{11} = x_{21} = \frac12$ with a total social welfare of $\frac{K}2 + \frac12$.
	
	Now increase bidder $2$'s budget to $B'_2 = 2$. The new \fppe is  $a_1 = \frac3K, s_2 = 1, x_{11} = \frac13, x_{21} = \frac23$ with a total social welfare of $\frac{K}3 + \frac23$. As $K\rightarrow \infty$, we lose $\frac16$ of the social welfare.
\eop\endproof

\begin{proposition}\label{prop:sw-value}
	In an \fppe, increasing a bidder $i$'s value for some good $j$ from $v_{ij}$ to $v'_{ij} > v_{ij}$ can decrease social welfare.
\end{proposition}
\proof{Proof.}
	Consider the following instance: 2 bidders, 1 good. We have values $v_{11} = K, v_{21} = 1$, and budgets $B_1 = \frac12, B_2 = 2$. The \fppe is $a_1 = \frac1K, a_2 = 1$ and $x_{11} = x_{21} = \frac12$, with a total social welfare of $\frac{k+1}2$. Now increase $v'_{21} = 2$. The new \fppe consists of $a_1 = \frac2K, a_2 = 1$, and $x_{11} = \frac14, x_{21} = \frac34$ with a social welfare of $\frac{k+3}{4}$. As $K\rightarrow \infty$ we lose $\frac12$ of the social welfare.
\eop\endproof

\subsection{Sensitivity Analysis}

We now investigate the sensitivity of \fppe to budget changes. An overview of our results is shown in Table~\ref{table:sens-results}. When adding $\Delta$ to the budget of a bidder, revenue can only increase, and at most by~$\Delta$. This shows that, in a sense, an \fppe is revenue (and thus paced-welfare) stable with respect to budget changes: the change in revenue is at most the same as the change in budget. In contrast to this, \citet{conitzer2021pacing} show that in an \sppe the revenue can change drastically, at least by a factor of $100$.

\begin{table}
\centering
\begin{tabular}{|l|c|c|}
	\hline 
	& Maximal Decrease & Maximal Increase \\ 
	\hline \hline
	Revenue (additive) & $0$ & $\Delta$ \\ 
	\hline 
	Social Welfare (relative)& $\frac{1-\Delta-\Delta^2}{1+\Delta}$ & $1+\Delta$\\ 
	\hline 
\end{tabular} 
\caption{Overview of sensitivity results. For revenue, the number is the upper bound on change in revenue as a result of increasing a bidder's budget by $\Delta$, i.e., $B'_i = B_i + \Delta$. For social welfare, the number is the upper bound on relative change in social welfare as a result of a relative increase in budget of $1+\Delta$, i.e., $B'_i = (1+\Delta)\cdot B_i$.}
\label{table:sens-results}
\end{table}

Due to the nature of multiplicative pacing, additive bounds for social welfare (such as the ones given for revenue) do not immediately make sense.\footnote{To see why, take any instance $(N, M, V, B)$ with budget-constrained bidders and compare it with an instance $(N, M, 2V, B)$ where the valuations are multiplied by 2. Changing a budget will yield the same allocation for both instances (and pacing multipliers are precisely a factor 2 off), but the change in social welfare will be twice as large in the second instance.} Therefore, we focus on sensitivity results for a \emph{relative} change in budget, leading to a \emph{relative} change in social welfare.

Our social welfare proofs rely on the fact that when a budget changes by factor $1+\Delta$, pacing multipliers can only change by at most a factor $1+\Delta$.

\begin{lemma}\label{lem:pm-delta}
	In an \fppe, changing one bidder's budget from $B_i$ to $B'_i = (1 + \Delta)B_i$ for $\Delta\ge 0$ yields a modified \fppe $\newfppet$ with $\alpha \le \alpha' \le (1+\Delta)\alpha$.
\end{lemma}
\proof{Proof.}
Fix the instance $(N, M, V, B)$ and let $\fppet$ be its \fppe. 
Let $B'_i = (1 + \Delta)B_i$ and $\newfppet$ be the \fppe for $(N, M, V, B')$. Note that $(\alpha, x)$ is a \bfpm for $(N, M, V, B')$, so $\alpha \le \alpha'$ by Corollary~\ref{cor:uniqueness}. For the other inequality, note that $(\frac{\alpha'}{1+\Delta}, x')$ forms a \bfpm for the original instance $(N, M, V, B)$. Indeed, all prices drop by exactly a factor $\frac{1}{1+\Delta}$, which means that with the same allocation $x'$, the spend for all bidders goes down by a factor $\frac{1}{1+\Delta}$ so no bidder exceeds their budget. By Corollary~\ref{cor:uniqueness} the pacing multipliers $\alpha$ of the \fppe on $(N, M, V, B)$  can only be higher, yielding $\alpha \ge \frac{\alpha'}{1+\Delta}$. Rearranging yields the claim.
\eop\endproof

To complete the proofs for social welfare, note that in an \fppe, pacing multipliers correspond to the bang-per-buck of buyers (i.e., the ratio between value and spend), so the bound in revenue change implies a bound in social welfare change.

\begin{proposition}\label{prop:sens-rev-ub}
In an \fppe, increasing a bidder $i$'s budget by $\Delta$, i.e., $B'_i = B_i + \Delta$, yields a revenue increase of at most $\Delta$.
\end{proposition}
\proof{Proof.}
	Fix the instance $(N, M, V, B)$, and let $B'$ be the budget profile where $B'_i = B_i + \Delta$ for some bidder $i$. Let $\fppet$ be the \fppe on $(N, M, V, B)$, and let $\newfppet$ be the \fppe on $(N, M, V, B')$. Since $\fppet$ is a \bfpm for the new instance, we have $\alpha' \ge \alpha$, that is, the new pacing multipliers are weakly higher than the old pacing multipliers. Let $S_+$ be the bidders for whom $\alpha'_k > \alpha_k$, and let $\text{Rev}_{+}^\text{old}$ be the revenue from them in $\fppet$. Since the pacing multipliers for all bidders in $S_+$ strictly increased, they must have had $\alpha_k < 1$, so by the definition of an \fppe they must have spent their entire budget and $\text{Rev}_{+}^\text{old} = \sum_{k\in S_+} B_k$. In the new \fppe, they cannot spend more than their budget, so $\text{Rev}_{+}^\text{new} \le \sum_{k\in S_+} B'_k \le \left(\sum_{k\in S_+} B_k\right) + \Delta = \text{Rev}_{+}^\text{old} + \Delta$.
    
What remains to be shown is that the revenue from the bidders $S_-$ with $\alpha'_k = \alpha_k$ cannot have gone up. If there were any goods that $S_+$ and $S_-$ were tied for, then after increasing the pacing multipliers of $S_+$, the prices of those goods increased and $S_+$ won all of them. Moreover, the prices of goods that $S_-$ as a set still wins have not changed. Thus $S_-$ is winning a subset of the goods they won previously at the same per-unit cost, hence their spend cannot have gone up.
\eop\endproof

Along with Proposition~\ref{prop:rev-budget}, this shows that when a bidder's budget increases by $\Delta$, $\text{Rev}^\text{new} - \text{Rev}^\text{old} \in [0, \Delta]$. It is not difficult to see that these extremes can both be attained: for the lower bound, increasing the budget of a non-budget-constrained bidder will not change the \fppe, hence revenue is unchanged. On the upper bound, take 1 bidder, 1 good, $v_{11} = 2\Delta, B_1 = \Delta$. Setting $B'_1 = B_i + \Delta$ will increase revenue by $\Delta$.

From Proposition~\ref{prop:sw-budget} below, we know that social welfare can decrease when we increase a bidder's budget. The following lemma bounds that loss. In the following, let $SW^\text{old}$ be the social welfare prior to changing the budget, and $SW^\text{new}$ be the social welfare after changing the budget.

\begin{proposition}\label{prop:sens-sw-lb}
	In an \fppe, changing one bidder's budget from $B_i$ to $B'_i = (1 + \Delta)B_i$ for $\Delta\ge 0$ yields $SW^\text{new} \ge \left(\frac{1-\Delta-\Delta^2}{1+\Delta}\right)SW^\text{old}$.
\end{proposition}
\proof{Proof.}
Let $i$ be the bidder with $B'_i = (1+\Delta)B_i$. Let $\fppet$ be the \fppe before the change, and let $\newfppet$ be the \fppe after the budget change. Let $S_p$ be the set of bidders who are paced in $\alpha'$, and let $S_1 = N \backslash S_p$ the unpaced bidders. We will lower-bound the new revenue from $S_p$ and $S_1$ separately.

For the bidders in $S_p$, they spend their entire budget in both $\fppet$ and $\newfppet$: we have that $\alpha \le \alpha'$ by Lemma~\ref{lem:pm-delta} and thus bidders that are paced in $\alpha'$ are also paced in $\alpha$. By the definition of \fppe that means they spend their entire budget. Moreover, by Lemma~\ref{lem:pm-delta}, their pacing multipliers cannot have gone up by more than $1+\Delta$, hence their bang-per-buck is at least $\frac{1}{1+\Delta}$ that in $\fppet$. Combining these statements, in $\newfppet$ bidders in $S_p$ spend at least as much as in $\fppet$, and their bang-per-buck is at least $\frac{1}{1+\Delta}$ times that in $\fppet$, hence their contribution to social welfare $SW^\text{new}_k \ge \frac{SW^\text{old}_k}{1+\Delta}$ for each $k\in S_p$, and therefore $SW^\text{new}_{p} \ge \frac{SW^\text{old}_{p}}{1+\Delta}$ .

For the set $S_1$ of bidders who are unpaced in $\alpha'$, their combined decrease in spend can be at most $\Delta\cdot B_i$: the total spend cannot have decreased by Proposition~\ref{prop:rev-budget}, bidder $i$'s spend increased by at most $\Delta\cdot B_i$, the paced bidders (excluding bidder $i$) in $\alpha'$ were also all paced in $\alpha$ so their spend stayed constant, hence the largest possible reduction in spend by unpaced bidders in $\alpha'$ is $\Delta\cdot B_i \le \Delta \cdot SW^\text{old}$. For unpaced bidders, spend equals contribution to social welfare, so we have $SW^\text{new}_{1} \ge SW^\text{old}_{1} - \Delta \cdot SW^\text{old}$.

Combining everything, we have $SW^\text{new} = SW^\text{new}_{p} + SW^\text{new}_{1} \ge \frac{SW^\text{old}_{p}}{1+\Delta} + SW^\text{old}_{1} - \Delta \cdot SW^\text{old} \ge \frac{SW^\text{old}}{1+\Delta} - \Delta\cdot SW^\text{old} = \left(\frac{1-\Delta-\Delta^2}{1+\Delta}\right)SW^\text{old}$.
\eop\endproof

\begin{proposition}\label{prop:sens-sw-ub}
	In an \fppe, changing one bidder's budget from $B_i$ to $B'_i = (1 + \Delta)B_i$ for $\Delta\ge 0$ yields $SW^\text{new} \le (1+\Delta)SW^\text{old}$.
\end{proposition}
\proof{Proof.}
Let $i$ be the bidder whose budget increases from $B_i$ to $(1+\Delta)B_i$. Let $\fppet$ be the \fppe before the change, and $\newfppet$ be the \fppe after the budget change. By Lemma~\ref{lem:pm-delta}, increasing a budget can only increase pacing multipliers. Let $S_+$ be the set of bidders whose pacing multiplier increased (for convenience excluding bidder $i$), let $S_-$ be the set who had pacing multipliers strictly lower than $1$ and that did not change, and let $S_1$ be the set of bidders who were unpaced in $\fppet$. Let $SW^\text{old}$ be the old social welfare, and $SW^\text{new}$ be the new one. Define $SW_i$, $SW_+$, $SW_-$, and $SW_1$ as the contribution to social welfare of bidder $i$, bidders $S_+$, $S_-$, and $S_1$ respectively.

We use extensively that at pacing multiplier $\alpha_k$, the spend $s_k = \alpha_k \cdot SW_k$.

For bidder $i$, we have $SW^\text{new}_i \le (1+\Delta)SW^\text{old}_i$: The pacing multiplier of $i$ can only have increased, so the bang-per-buck can only have decreased. Spend increased at most by $(1+\Delta)$, bang-per-buck was at most the same, hence $SW$ cannot exceed $1+\Delta$ more.

For bidders $S_+$, we have $SW^\text{new}_+ < SW^\text{old}_+$: Their spend cannot have increased as they spent their budget completely in $\fppet$. Meanwhile, their bang-per-buck strictly decreased due to increasing pacing multipliers.

For bidders $S_-$, we have $SW^\text{new}_- = SW^\text{old}_-$: Since they were and are paced, they must spend their entire budget. Since their pacing multiplier has not changed, their bang-per-buck stayed the same. Thus their contribution to $SW$ stayed the same.

Finally, for bidders $S_1$, we have $SW^\text{new}_1 \le SW^\text{old}_1$:
First note that the total spend of bidders in $S_- \cup S_1$ cannot have increased; their pacing multipliers stayed the same, so this would require them to win strictly more of some good they were previously sharing with a bidder not in $S_- \cup S_1$, but they are no longer winning \emph{any} such goods. Since we have seen that the spend of bidders in $S_-$ stayed the same, the spend of bidders in $S_1$ cannot have increased. Since their bang-per-buck is $1$, their $SW$ cannot have increased.

Summing over all groups: $SW^\text{new} = SW^\text{new}_i + SW^\text{new}_+ + SW^\text{new}_- + SW^\text{new}_1 \le (1+\Delta)SW^\text{old}_i + SW^\text{old}_+ + SW^\text{old}_- + SW^\text{old}_1 \le (1+\Delta)SW^\text{old}$.
\eop\endproof

\section{Algorithms via Convex Programming}
We now turn to computing the \fppe corresponding to an instance and adapt a well-known method for competitive
equilibria. Solutions to the Eisenberg-Gale convex program for Fisher markets
with quasi-linear utilities correspond exactly to \fppe in our setting.
\citet{cole2017convex} give the following primal and dual convex programs  for computing a solution to a Fisher market with
quasi-linear utilities. 

\noindent\begin{minipage}{.50\linewidth}
  \begin{gather}
    \max_{x\geq 0,\delta \geq 0, u} \quad \sum_{i} B_i\log(u_i) - \delta_i \nonumber\\
    u_i \leq \sum_j x_{ij} v_{ij} + \delta_i, \forall i  \label{eq:utility}\\
    \sum_i x_{ij} \leq 1, \forall j  \label{eq:supply}
  \end{gather} 
\end{minipage}
\noindent\begin{minipage}{.50\linewidth}
\begin{equation*}
  \begin{gathered}
    \min_{p\geq 0,\beta \geq 0} \quad \sum_{j} p_j - \sum_i B_i\log(\beta_i) \\
    \forall i, p_j \geq v_{ij}\beta_i \\
    \beta_i  \leq 1
  \end{gathered} 
\end{equation*}
\end{minipage}

We show the primal convex program on the left, which we denote by {\mcp}, and the corresponding
dual convex program on the right. The variables $x_{ij}$ denote the amount of item $j$ that bidder $i$ wins.
The leftover budget is denoted by $\delta_i$, it arises from the dual program:
it is the dual variable for the dual constraint $\beta_i \leq 1$, which
constrains bidder $i$ to paying at most a cost-per-utility rate of $1$.

The dual variables $\beta_i$, $p_j$ correspond to constraints \eqref{eq:utility}
and \eqref{eq:supply}, respectively. They can be interpreted as follows:
$\beta_i$ is the inverse bang-per-buck: $\min_{j :
  x_{ij}>0}\frac{p_j}{v_{ij}}$ for buyer $i$, and $p_j$ is the price of good~$j$.

We now show via Fenchel duality that {\mcp} computes an {\fppe}. 
Informally, the result follows because $\beta_i$ specifies a single utility rate per bidder,
duality guarantees that any item allocated to $i$ has exactly rate
$\beta_i$, and thus since {\mcp} is known to compute a competitive equilibrium
we have shown that it computes an \erce. Theorem~\ref{thm:erce iff fppe} then
gives the result.
\begin{theorem}
  An optimal solution to {\mcp} corresponds to an \fppe
  with pacing multiplier $\alpha_i = \beta_i$ and allocation $x_{ij}$,
  and vice versa.
  \label{thm:convex_program_fppe}
\end{theorem}

\proof{Proof.}
  We start by listing the primal \kkt conditions:

  \begin{multicols}{2}
  \begin{enumerate}
  \item $\frac{B_i}{u_i}  = \beta_i \Leftrightarrow
    u_i  = \frac{B_i}{\beta_i}$ \label{kkt:utility}
  \item $\beta_i \leq 1$ \label{kkt:quasi_linear}
  \item $\beta_i \leq \frac{p_j}{v_{ij}}$ \label{kkt:utility_rate}
  \item $x_{ij}, \delta_i, \beta_i, p_j \geq 0$ \label{kkt:nonnegative}
  \item $p_j>0 \Rightarrow \sum_i x_{ij} = 1$ \label{kkt:full_allocation}
  \item $\delta_i>0 \Rightarrow \beta_i = 1$ \label{kkt:v_pos}
  \item $x_{ij}>0 \Rightarrow \beta_i=\frac{p_j}{v_{ij}}$ \label{kkt:x_pos}
  \end{enumerate}
  \end{multicols}
  \newcommand{\kktref}[1]{\kkt condition \eqref{#1}}
  \newcommand{\kktrefs}[1]{\kkt conditions \eqref{#1}}

  It is easy to see that $x_{ij}$ is a valid allocation: {\mcp} has the exact
  packing constraints. Budgets are also satisfied (here we may assume $u_i>0$
  since otherwise budgets are satisfied because the buyer wins no goods): by
  \kktrefs{kkt:utility} and \eqref{kkt:x_pos} we have that for any good $j$ that
  bidder $i$ is allocated part of:
  \begin{align*}
    \frac{B_i}{u_i} = \frac{p_j}{v_{ij}} \Rightarrow 
    \frac{B_iv_{ij}x_{ij}}{u_i} = p_j x_{ij}\,.
  \end{align*}
  If $\delta_i = 0$ then summing over all $j$ gives
  \begin{align*}
    \sum_j p_j x_{ij} = B_i\frac{\sum_jv_{ij}x_{ij}}{u_i} = B_i\,.
  \end{align*}
  This part of the budget argument is exactly the same as for the standard
  Eisenberg-Gale proof~\citep{nisan2007algorithmic}. Note that constraint
  \eqref{eq:utility} of {\mcp} always holds exactly since the objective is strictly
  increasing in $u_i$. Thus $\delta_i=0$ denotes full budget expenditure. If $\delta_i>0$
  then \kktref{kkt:v_pos} implies that $u_i = B_i$. 
  This gives:
  \begin{align*}
    \sum_j p_j x_{ij} + \delta_i = B_i\frac{\sum_jv_{ij}x_{ij}}{u_i} + \frac{B_i}{u_i}\delta_i = B_i\,.
  \end{align*}
  This shows that $\delta_i>0$ denotes the leftover budget.

  If bidder $i$ is winning some of good $j$ ($x_{ij}>0$) then \kktref{kkt:x_pos}
  implies that the price on good $j$ is $\alpha_iv_{ij}$, so bidder $i$ is
  paying their bid as is necessary in a first price auction. Bidder $i$ is also
  guaranteed to be among the highest bids for good $j$: \kktrefs{kkt:utility_rate} and \eqref{kkt:x_pos} guarantee that  $\alpha_iv_{ij}=p_j\geq
  \alpha_{i'}v_{i'j}$ for all $i'$.

  Finally, each bidder either spends their entire budget or is unpaced:
  \kktref{kkt:v_pos} says that if $\delta_i>0$ (that is, some budget is leftover)
  then $\beta_i=\alpha_i=1$, so the bidder is unpaced. 
  
  Now we show that any \fppe satisfies the \kkt conditions for {\mcp}. We set
  $\beta_i = \alpha_i$ and use the allocation $x$ from the \fppe. We set $\delta_i=0$
  if $\alpha < 1$, otherwise we set it to $B_i - \sum_{j}x_{ij}v_{ij}$. We set
  $u_i$ equal to the utility of each bidder.
  \kktref{kkt:utility} is satisfied since each bidder either gets a utility rate
  of $1$ if they are unpaced and so $u_i = B_i$ or their utility rate is
  $\alpha_i$ so they spend their entire budget for utility $B_i/\alpha_i$.
  \kktref{kkt:quasi_linear} is satisfied since $\alpha_i\in [0,1]$.
  \kktref{kkt:utility_rate} is satisfied since each good bidder $i$ wins has
  price-per-utility $\alpha_i=\frac{p_j}{v_{ij}}=\beta_i$, and every other good
  has a higher price-per-utility.
  \kktrefs{kkt:nonnegative} and \eqref{kkt:full_allocation} are
  trivially satisfied by the definition of \fppe. \kktref{kkt:v_pos} is satisfied
  by our solution construction. \kktref{kkt:x_pos} is satisfied because a bidder
  $i$ being allocated any amount of good $j$ means that they have a winning bid,
  and their bid is equal to $v_{ij}\alpha_i$.
\eop\endproof

This shows that we can use {\mcp} to compute an \fppe. \citet{cole2017convex} show that {\mcp} admits rational equilibria, and thus an \fppe can be computed in polynomial time with the ellipsoid method as long as all inputs are rational. Furthermore, the relatively simple structure of {\mcp} means that it can easily be solved via standard convex-programming packages such as CVX~\citep{gb08,grant2008cvx} or scalable first-order methods.

\citet{borgs2007dynamics} also gave a convex program. However, that convex program is a set of feasibility constraints (in fact, those constraints correspond to the \kkt conditions of \mcp{}). Because of this, their convex program cannot easily be solved via first-order methods, which are suitable for very large-scale problems. Similarly, online convex optimization can potentially be applied to \mcp{} (or its dual) for solving online variants of \fppe, whereas one cannot do this for the convex program of \citet{borgs2007dynamics}. See, e.g., \citet{gao2020first,gao2020infinite} for recent work that builds on our results in this direction.

The equivalence between solutions to {\mcp} and \fppe provides an alternative
view on many of our results. Since Theorem~\ref{thm:convex_program_fppe} can be
shown directly via Fenchel duality, it allows us to
prove via duality theory that \fppe corresponds to \erce, and that \fppe always
exists.\footnote{{\mcp} is easily seen to always be feasible and the feasible set is
compact. Thus {\mcp} always attains an optimal solution. By
Theorem~\ref{thm:convex_program_fppe} that solution will be an \fppe.}

\section{Experiments}
In previous sections, we have shown that \fppe have many satisfying theoretical properties, such as guaranteed existence, uniqueness, and polynomial time computability. However, the empirical question remains: if participating in a single first price auction has such bad properties that industry shifted away from them in the past, does that assessment change when the auctions are within a budget management system? Additionally, while \fppe have nice properties compared to \sppe, how do the equilibria for first price auctions compare to those for second price auctions for typical objectives like social welfare and for typical instances? We investigated the properties of \fppe via numerical simulations\footnote{There is a public notebook on synthetic data for these experiments available on Github at \url{https://github.com/facebookresearch/fppe}\,.}
on realistic instances generated from data collected at Facebook and Instagram. We aimed to answer two concrete questions: (1) Under \fppe, how high is bidder regret for reporting truthfully? (2) How does \fppe compare to \sppe in terms of revenue and social welfare?

\subsection{Data}
Similarly to instances generated by \citet{conitzer2021pacing}, we construct realistic instances from real-world auction markets in two steps. We first take bidding data for a region during a period and use it to create $n$~bidders and $m$~goods. To get the $n$~bidders, we identify the top $n$ advertisers that participate in the most auctions in that period in that region. 
Each of those advertisers will map to a buyer in the final instance. As an intermediate step, we define the goods in the instance as the real-life auctions that include at least one of the $n$ bidders. In this intermediate step, we set the bid in each bidders-auction pair to be the value of the bidder for the corresponding good.\footnote{The campaigns used in this study were themselves frequently budget paced. In the logs, we have both the original bid and the pacing multiplier that were used. We use the original bid in this step.}
In a second step and to complete the construction, we reduce the size of the market so we can compute results more efficiently and thoroughly answer the questions of interest. To achieve that, we cluster the real-world auctions (the goods of the intermediate step) into a smaller number of goods. We apply the $k$-means algorithm using the $n$-dimensional vector of values for that good as features. The goods in the final instances are the resulting auction clusters. Each bidder valuation for a (cluster-level) good is set to the average of their valuations for goods in the cluster.
To generate the budgets, we set it equal to the expected value that the bidder would receive in a uniform random allocation of goods to bidders, i.e., $B_i = \frac{1}{n}\sum_j v_{ij}$. The motivation for this is that it leads to a good mixture of budget-constrained and unconstrained bidders, since in aggregate this constrains the sum of prices to be the sum of average valuations, whereas it would be the sum of maximum valuations if every buyer was unconstrained. The distribution of budget-constrained and unconstrained bidders is also similar to that of \citet{conitzer2021pacing}, despite a slightly different process for setting budgets.\footnote{Scaling budgets might yield different results. When, for example, we scale budgets up by a factor 2, the fraction of budget-constrained buyers changes and the quantitative results change as a result. We refer the reader to the publicly available notebook mentioned earlier to try out different budget levels in the synthetic experiments and explore the effects.}
 
We collected bidding data from Facebook and Instagram for 7 days
for auctions in a chosen country.
For each day and each of those two platforms, we collect the bids of the top $n$ advertisers (where $n \in \{6, 8, 10, 12, 14\}$) with the most impressions that day. We consider $m \in \{10, 20, 30\}$ goods  where each good is a representative cluster of auctions. In the end, each combination of platform, day, number of bidders and number of goods becomes an instance of our study, resulting in a total of $2\times 7 \times 5 \times 3 = 210$ instances. To compare to \sppe, we take the same data and consider a set of instances with $\{3, \ldots , 8\}$ bidders and $\{4, \ldots , 8\}$ goods, for a total of $2\times 7 \times 6 \times 5 = 420$ instances. The \sppe instances are smaller because equilibria are harder to compute.

While we present our results on moderate-size instances, we also performed some tests on larger synthetic data to verify that \mcp can indeed be solved efficiently. 
For these tests, we used the Mosek interior-point method~\citep{dahl2021primal}. We found that for uniform random valuations, instances with 4000 buyers and 4000 items can be solved in about 12 minutes, and runtime scaled approximately linearly with the size of the valuation matrix.
In addition to these encouraging experiments via conic solvers, it was recently shown that {\mcp} can be solved using first-order methods~\citep{gao2020first}, which should enable further scalability, at the cost of slightly more inaccurate solutions.

\subsection{Incentive Properties}

We start by numerically analyzing the incentive properties of \fppe. Our analysis will focus on two types of incentives: first, ex-post incentives to shade bids due to the first price auction rule, when other buyer's bids are held constant, and second, ex-ante incentives to misreport in order to shift the \fppe outcome itself when viewed as a mechanism. For both types of incentives, we will focus on the proxy-bidder setting, where advertisers typically report a value per conversion and budget, while the valuation for individual auctions is determined by the platform as the value per conversion times the conversion rate. Thus, the space of possible misreports is substantially reduced: buyers can misreport their overall value per conversion and/or their budget (as opposed to individually shading or misreporting in each auction).

\paragraph{Analysis of Regret.} First we look at the ex-post regret that each bidder has in \fppe as compared to being able to unilaterally deviate to a different pacing multiplier, while keeping the \fppe multipliers fixed for all other bidders.
The regrets resulting from our computational study are shown in Figure~\ref{fig:exp regrets}. 
The figure shows summary statistics over relative ex-post regret, which is the fraction of utility for the best-response pacing multiplier that is lost if the bidder reports truthfully. Each data point represents the relative regret that one of the bidders experienced in the \fppe of each instance; we generate a data point for each bidder in the instance. 
The middle line in each box (which is at $0.00$ for all breakdowns) is the median; the lower and upper edges of the boxes represent the first and third quartiles; the lines extending from the box show the rest of the distribution except for outliers that are determined by a function of the inter-quartile range; and the dots represent individual outliers outside that range. 

\begin{figure}[t]
  \centering
	\includegraphics[width=\columnwidth]{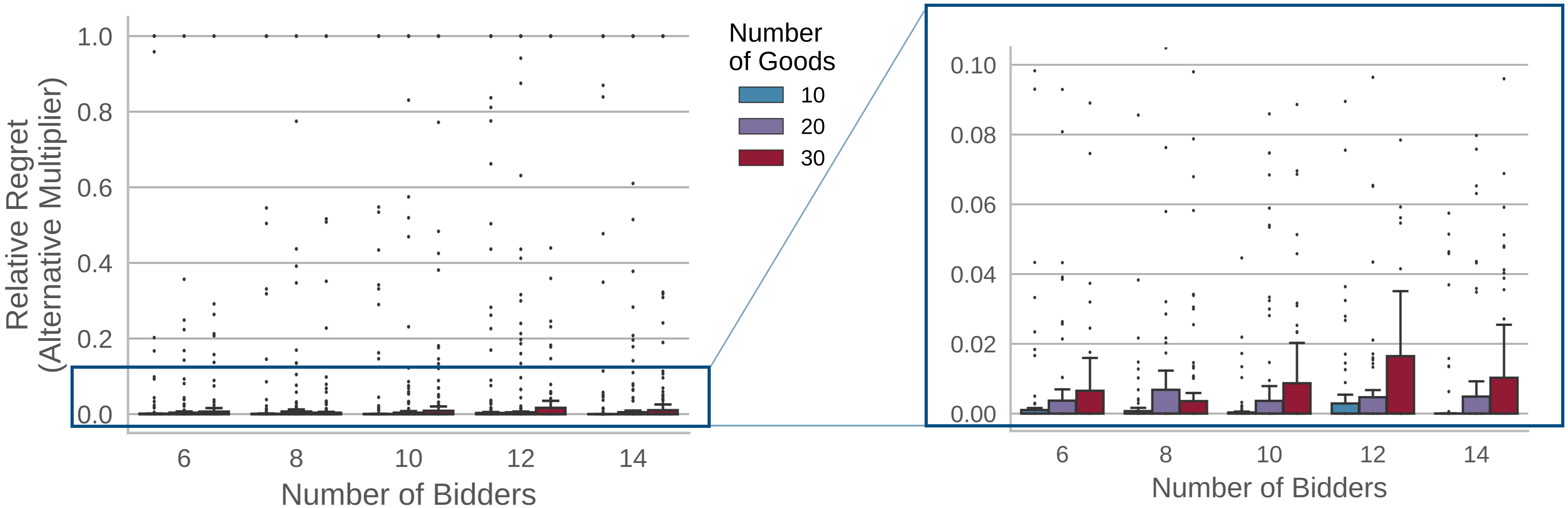}
    \caption{Summary statistics of relative ex-post regret in \fppe (ratio of best-response utility keeping competitors' bids fixed to utility under the \fppe). The plot on the right is a zoomed-in version of the plot on the left.}
  \label{fig:exp regrets}
\end{figure}

The results show that across all breakdowns of number of bidders and number of goods, the relative regret is small. In all cases, the median is at $0\%$ relative regret, and the highest upper quartile represents less than $2\%$ relative regret, so bidders across the board do almost as well in \fppe as they would if they could optimally pick their pacing multiplier (given fixed competition). We do see that the benefit to deviating increases somewhat with the number of goods (instances with 10 goods show consistently lower regret than instances with 20 or 30 goods), while there does not seem to be a strong dependence on the number of bidders.

To put the previous analysis in perspective, Appendix~\ref{sec:additional_realistic_experiments} presents results for an unrealistic model where bidders can optimally set their paced bids for each individual auctions independently.%
\footnote{This deviation model is unrealistic in ad markets because a bidder would have to change the platform's estimated conversion rate for each user in such a way that it leads precisely to the right paced bid for each auction. The reason we consider deviations of this kind is to show that even under unreasonably powerful assumptions of what a bidder can do, the relative regret is still low.} 
Even under these conditions, across all breakdowns, $75\%$ of bidders experience less than $30\%$ relative regret, and as the number of bidders grows (compared to the number of goods), the majority of bidders in instances have $0\%$ relative regret.

\paragraph{Misreporting to Influence Equilibria.}

After providing evidence that bidders have low ex-post regret if they had chosen the best pacing multiplier while keeping the competition fixed, we look at whether bidders can
influence the \fppe outcome itself by submitting conveniently chosen (but potentially wrong) parameters to the \fppe mechanism. 
We assume that a focal bidder can misreport by scaling the valuations and budget by scalars chosen from the set $\{0.1,0.2,\ldots,1.1\}^2$. In other words, the bidder would select a tuple $(\lambda_v, \lambda_b)$, and submit the true valuations multiplied by $\lambda_v$ and the true budget multiplied by $\lambda_b$. Note that the bidder does not have the ability to scale values differently for different auctions (but see Appendix~\ref{sec:additional_realistic_experiments} for an analysis of that case). All other bidders' bids and budgets remain unchanged.
Figure~\ref{fig:exp misreport} shows the relative regret with respect to the highest utility at an \fppe amongst the possible ways to misreport. In almost all cases there is no or negligible gain. We note also that the  incentive to misreport is even lower than the incentive for ex-post shading discussed in the previous section (of the order of 0.001 vs.\ of the order of 0.01 relative regret). This is likely because, in most cases, the ideal ex-post shading for a single buyer leads to budget infeasibility for other buyers, which in turn makes that outcome infeasible.

\begin{figure}[t]
\centering
\includegraphics[width=\columnwidth]{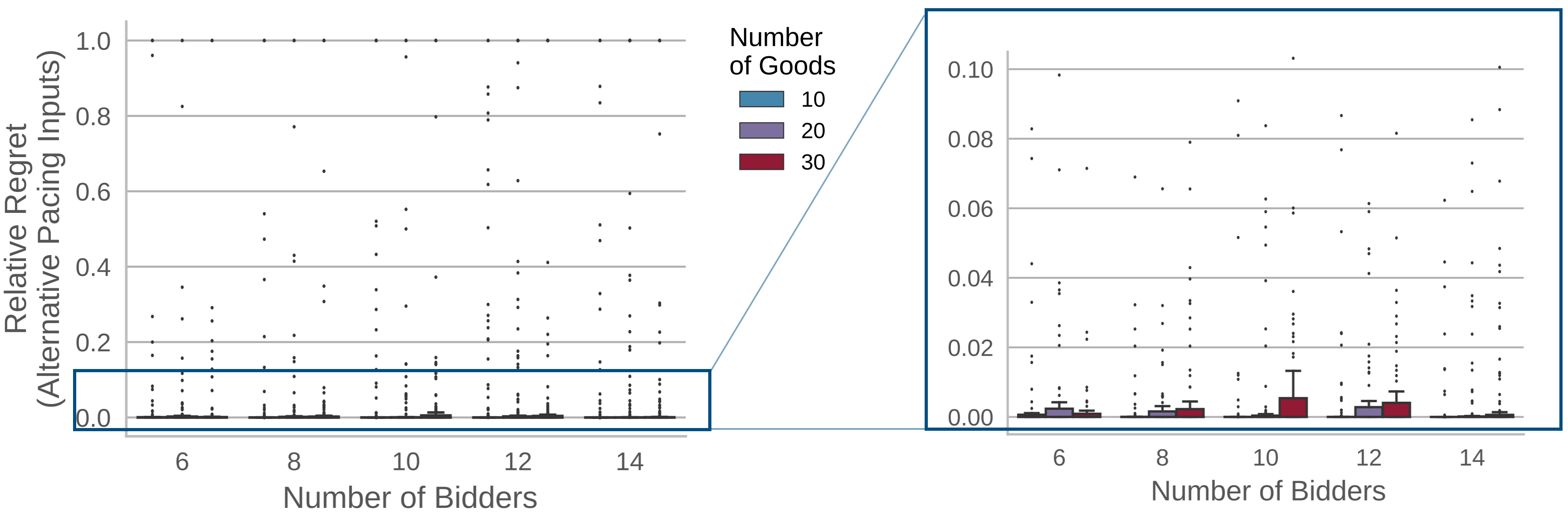}
\caption{Summary statistics when bidders can misreport their value per conversion or budget to manipulate the \fppe.}
 \label{fig:exp misreport}
\end{figure}

\paragraph{Regret as a Function of the Tightness of Budget Constraints.} 
While the majority of bidders in these experiments have surprisingly low regret under \fppe, there are some outliers that have high relative regret. To investigate where this variability comes from, we look at the regret that a bidder experiences as a function of their pacing multiplier at equilibrium. Intuitively, a bidder who has a low pacing multiplier gets a high bang-per-buck for each impression they buy. This implies that underbidding loses value faster than it saves cost, whereas a bidder with a pacing multiplier equal to 1 does not get any utility in each auction won. Our computations indicate that bidder budgets, pacing multipliers, and regrets all have a positive association.

\begin{figure}[t]
  \centering 
	\includegraphics[width=0.49\columnwidth]{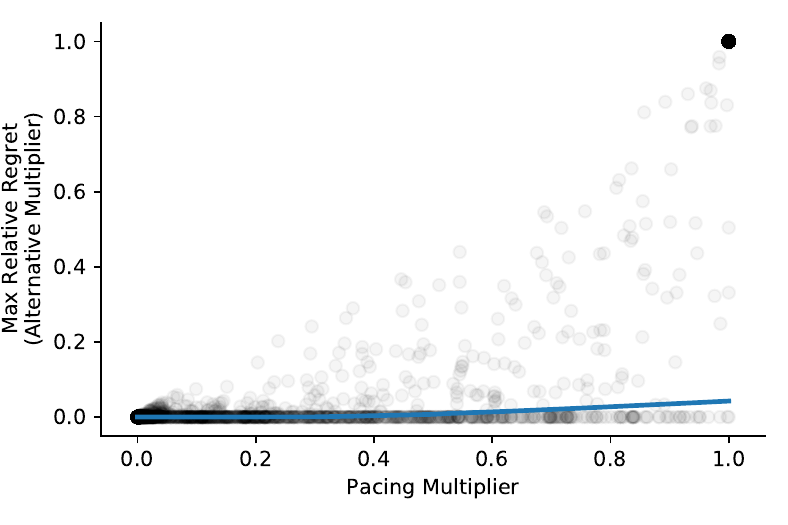}
	\includegraphics[width=0.49\columnwidth]{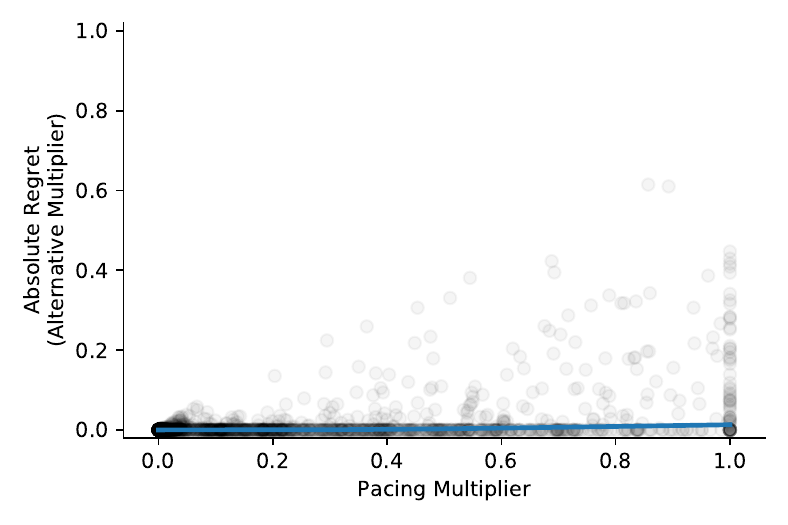}
	\caption{Regret as a function of \fppe pacing multiplier. Each point represents the regret of an individual bidder; the curve summarizes the points using a LOWESS regression. The left and right panels display relative and absolute regret, respectively.}
  \label{fig:regret_vs_multiplier}
\end{figure}

Figure~\ref{fig:regret_vs_multiplier} plots regrets for these instances as a function of the pacing multiplier, where each point represents a bidder of a pacing instance. The left panel is a scatter plot of the relative regret, defined as the fraction of utility that a bidder loses in the \fppe compared to the best-response pacing multiplier. Since the regret for the majority of data points is approximately zero, we also plot a curve corresponding to a LOWESS regression to capture the `typical' behavior.\footnote{The LOWESS (locally weighted scatter plot smoothing) regression is a non-parametric method that is frequently used for data analysis \citep{cleveland1979robust}. For instance, it is the default way to run a regression with smoothed data in R using the function {\sc geom\_smooth}. At each point, it fits a local low-degree polynomial where the input data is weighted by distance to that point.} 
The typical relative regret increases when the pacing multiplier increases because of the smaller relative frequency of points with no regret and a noticeable mass of points located at $(1, 1)$. At first glance, this implies that when the pacing multiplier exceeds $0.4$, the typical regret grows linearly from $0$ to about $0.1$. The points at $(1,1)$ represent unpaced bidders that receive no utility in the \fppe; while there is no numerical error in the calculation, relative regrets do not fully depict the situation. For example, if a bidder gets impressions they value at $\$100$ and pay $\$100$ in the \fppe, while a best-response pacing multiplier may make them pay $\$98$ to get $\$99$ worth of impressions, their relative regret is $1$, but in absolute terms the outcomes are very similar. Consequently, we complement the analysis with an absolute notion of regret that captures the utility gained under a best response. To more easily compare bidders against each other, we normalize the utility gained by the total value of goods received for the best-response pacing multiplier. The results are displayed on the right panel of Figure~\ref{fig:regret_vs_multiplier}. The point at $(1,1)$ has disappeared, and the best-fit line for bidders with pacing multiplier equal to $1$ show a regret down to roughly $1.5\%$. From the distribution of instances (the black dots), it can also been seen that small pacing multipliers are more common. So while positive deviations do exist, on average the benefit is small.

\subsection{Revenue and Social Welfare}

We now compare revenue and social welfare under \fppe and \sppe, as shown in Figure~\ref{fig:exp sppe comparison}. The left panel shows the cdf
of the ratio of \fppe revenue to that of \sppe. The right panel is similar but with social welfare. We see that \fppe
revenue is always higher than \sppe revenue, though both coincide for about
$30\%$ of instances, and almost never more than $4.5$ times as high. For social welfare
we find that, perhaps surprisingly, neither solution concept is dominating, with
most instances having relatively similar social welfare under either solution concept,
though \fppe does slightly better. There are two caveats to keep in mind for
these results: (a) we did not compute the social-welfare-maximizing \sppe so
it is possible that there is a better one (although this is highly unlikely
given that \citet{conitzer2021pacing} find that most instances admit a single equilibrium); (b) many bidders are budget
constrained in the \fppe of our setting, and so these insight might not translate to
cases where many bidders are not budget constrained (see Appendix~\ref{app:revenue_welfare} for
statistics on multipliers in the two solution concepts). These experiments show
that an \fppe is not necessarily worse than an \sppe with respect to social welfare (at least with
nonstrategic bidders), while potentially having a significantly higher revenue.

\begin{figure}[t]
  \centering
	\includegraphics[width=0.49\columnwidth]{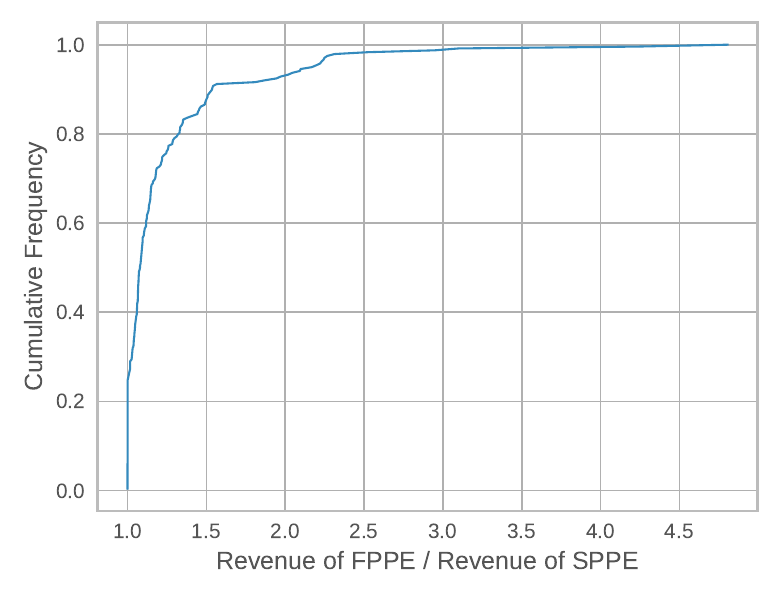}
	\includegraphics[width=0.49\columnwidth]{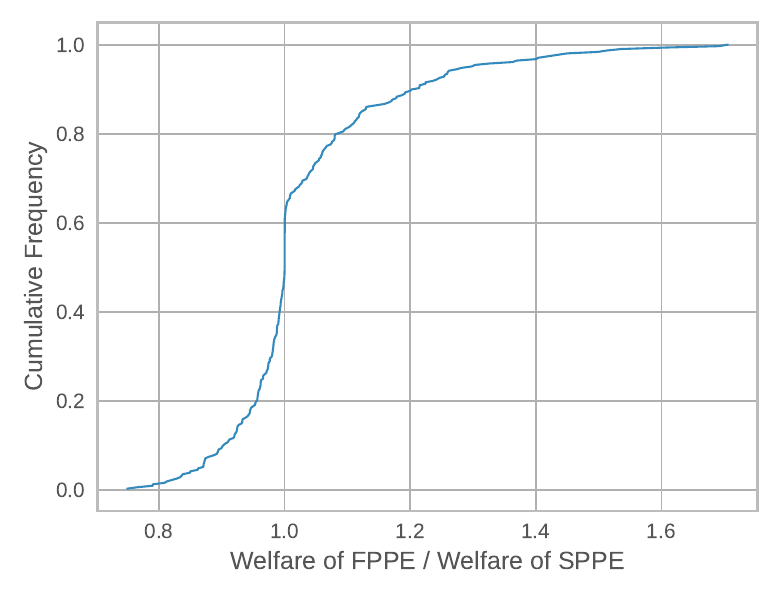}
	\caption{Summary statistics over the \fppe\ / \sppe ratio of revenue (left) and
    social welfare (right). }
  \label{fig:exp sppe comparison}
\end{figure}

\subsection{Robustness}

In order to test the robustness of our results, we ran the same set of experiments on synthetic data (see Appendix~\ref{sec:synthetic_experiments} for details). Qualitatively the results on instances with values drawn from a unif(0,1) distribution are similar, though regret tends to be higher in the synthetic data sets. One hypothesis is that values in the realistic data are correlated, for example because certain users are overall more likely to interact with ads, and this results in stronger competition for high-value auctions.

\section{Discussion of Modeling Assumptions}
\label{sec:assumptions discussion}

\noindent

Our model captures several salient features present in contemporary real-world Internet advertising systems: budget management is done via a multiplicative scalar after buyers report a value per conversion and a budget. Due to the deterministic valuations, we can allow general valuation structures as opposed to independent and/or identical stochastic valuations.
That said, there are also a number of practical issues that we do not capture. We discuss these limitations below, as well as potential ideas on how to relate them to our model.

\noindent{\em Stochastic valuations}: In line with the information available by the setting explored here, our model assumes that all valuations are known ahead of time, and consequently we have a completely determined market. The literature generally considers models in which user types are realized according to some sort of stochastic process, and each buyer's valuation is sampled according to a distribution that depends on the user type. Other authors have modeled the buyer's valuations as being drawn independently each time a user arrives~\citep{balseiro2015repeated,balseiro2017dynamic,balseiro2017budget}. While such models address the stochastic aspects of the process, they fail to capture correlations. Conversely, our model does capture the correlation, but fails to capture the stochastic aspect. In practice, we would ideally want a model that captures both, but such models are more cumbersome to analyze. Nonetheless, our model can be viewed as a discrete approximation to this kind of process.

\noindent{\em Dynamics}: In practice, platforms would typically have a model of what will occur: they know which advertisers are active in the system, their stated budgets and values per conversion, as well as the distribution of users. However, the platform does not know a priori which subset of users will actually arrive on a given day, and thus will not know which exact auctions will materialize. Instead, these auctions materialize sequentially as the day unfolds and users visit the platform. In practice, this  uncertainty is handled by a control algorithm that tunes the pacing multiplier to the right spending rate with the goal of exhausting the budget at the end of the planning horizon.
While we do not model these dynamics, our model can be related to dynamic markets in several ways.
First, our model captures the \emph{ideal} outcome: when each proxy bidder operates an independent control process on a pacing multiplier, the hope is that these pacing multipliers stabilize in a steady state. Our model describes what that steady state should be.
Secondly, in practice the platform would be able to create a good model of the overall composition of the market. Our computational results can be used to compute static pacing multipliers on such a market model. These static pacing multipliers can be used directly in the dynamic market. They can also be used to seed the control algorithm which further refines the pacing multiplier as the market progresses.
Thirdly, showing that pacing equilibria is directly connected to the \eg convex program opens the door to applying \emph{online convex optimization} techniques for converging to \fppe in an online fashion. While the exact details of how to achieve this are beyond the scope of this paper, we believe that it could be achieved using the online variant of regularized stochastic dual averaging~\citep{xiao2010dual}. In recent follow-up work to our paper, \citet[Theorem 7]{gao2020infinite} showed how to do this in a related Fisher market model.

\noindent{\em Multiple slots}: In practice each auction may allocate more than a single item. Typically a visit of a user triggers a multi-item auction (for example, several ad slots are auctioned off simultaneously both for Facebook and Instagram feed ads and for Google search auctions). Similarly to other authors, we abstract away this aspect of the problem, and assume that each item can be sold in isolation~\citep{balseiro2015repeated,balseiro2017budget,conitzer2021pacing,balseiro2017dynamic}.
Modeling multiple slots would require substantial changes in the analysis of the auction mechanism, and it would be harder to leverage the relationship to market equilibria.

\section{Conclusion}
\label{sec:conclusion}

In an advertiser platform, we must continually remember that the auction is only a small piece of a larger, complex system; when setting auction rules it is not just the properties of the auction that matter, but the platform's ensuing  aggregate behavior. 
In the case of feed-based advertising platforms, advertisers commonly control their campaigns through targeting criteria and budgets, with no ability to directly control individual auctions. This substantially changes the advertising platform’s system design problem compared to other settings.
In this paper, we have seen that the steady-state properties of a pacing-based budget management system for such scenarios are in fact quite good when a first price auction is used to sell each impression. 
We have showed that first price auctions enjoy several properties that second price auctions do not: uniqueness, monotonicity, and computability, among others. At the same time, \sppe enjoys other advantages, such as no static shading incentives (although our experiments on real ad-auction data suggest that this may not be a major concern in \fppe either).

In retrospect, the benefits of using a first price auction are not surprising. In simple settings, second price auctions win most theoretical head-to-head contests over first price auctions; however, it is well-known that the luster of the second price auction fades as it is generalized to a Vickrey-Clarke-Groves (VCG) auction, so much that VCG earned the title ``lovely but lonely'' for its lack of use in practice~\citep{ausubel2006lonely}. Indeed, some of the strengths of an advertising platform based on first price auctions are analogous to those seen in other complex auction settings~\citep{milgrom2004theorytowork}---uniqueness of equilibria, relation to competitive equilibria, core membership, shill-proofness, false-name-proofness, etc.---suggesting that first price auctions may, in fact, have a serious role to play in today's advertising platforms.

Taking a broader perspective, we recognize that there is generally no one-size-fits-all answer to the design of a budget management system, since no class of systems strictly dominates the others, and trade-offs remain. This is readily evidenced by the fact that different large platforms have adopted different solutions, both in terms of which budget smoothing method to adopt, and which auction format to run. Many factors will be weighed by any real-world system. For example, comparing multiplicative pacing systems to throttling systems, the former might be appealing because they encapsulate bidders’ natural incentive to shade their bids, whereas the latter make it easier for bidders to get a somewhat more representative sample of available impressions. Similarly, we see different solutions in use in the real world between first price and second price mechanisms. What we have done in this paper is to put forward a framework where the two auction formats can be compared from an end-to-end perspective that captures several design choices of real-world advertising platforms.

%
%
%

\begin{APPENDICES}
\section{Supplemental Experiments on Realistic Instances}
\label{sec:additional_realistic_experiments}

In addition to the regret analysis in the body of the paper (which is with respect to the ex-post best pacing multiplier given a fixed competition), in this section we also consider a much stronger model of deviations where bidders can optimally place individual bids in auctions corresponding to each good. This allows us to put the regret analysis discussed in the body of the paper in perspective compared to what can be considered as an upper bound on regret.
Recall that the platform computes the effective bid as 
conversion rate $\times$ value per conversion $\times$ pacing multiplier
for the bidder-good pair. Although in practice the advertiser only sets the value per conversion, in this section we assume that bidders can choose the optimal bid by influencing the platform's conversion rate estimates for the bidder-good pair. Consequently, the bidder can effectively control the individual bid placed in each of the auctions. 

The resulting relative regret is shown in Figure~\ref{fig:bid-regret}. As in the body, the middle line in each box represents the median, the lower and upper ends of the box show the first and third quartiles, the lines extending from the box show outliers within 1.5 times the inter-quartile range, and the dots represent individual outliers outside that range.
For all parameter settings (number of bidders, and number of goods) $75\%$ of bidders in all instances have less than $30\%$ relative regret. Considering how much control bidders have on the deviations in this analysis, it is rather surprising that in virtually all instances, the \fppe outcome yields more than $70\%$ of the optimal utility achievable. Moreover, we can see that the distribution of regret (in particular, the three quartiles) gets smaller as the competition in the market increases; e.g., when the number of goods decreases or the number of bidders increases. 
For 12 or more advertisers and 10 users, the median regret is $0$ indicating that more than half of all the buyers get their optimal utility in the equilibrium.

\begin{figure}[b]
    \centering
    \includegraphics[width=0.5\textwidth]{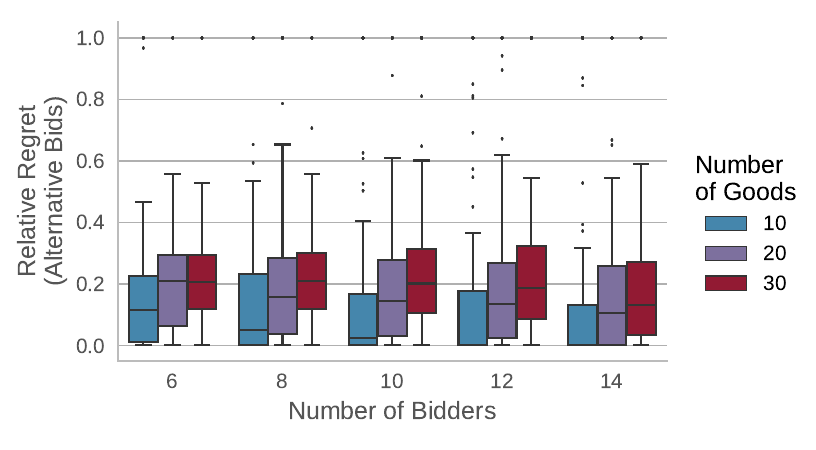}
    \caption{The relative regret of advertisers in \fppe compared to a hypothetical ability to directly control bids in each individual auction.}
    \label{fig:bid-regret}
\end{figure}

\section{Supplemental Experiments on Synthetic Instances}
\label{sec:synthetic_experiments}

In the body of the paper we have shown that based on realistic data, advertisers have low regret in \fppe compared to alternative inputs or pacing multipliers, and that \fppe generate more revenue than \sppe, with roughly similar welfare. To understand whether these results are specific to the data that we used, or if these results continue to hold for very different data, we repeat the analysis on synthetic data where bidders have valuations drawn from a uniform distribution. Since realistic data may have correlated valuations and valuation distributions with a long tail, this synthetic data represents a significant departure from realistic data. While the data is quite different, this is a robustness check where we find qualitatively similar insights, as described below.

We generated instances according to the
complete-graph model of \citet{conitzer2021pacing}. In this model, every bidder
is interested in every good, and each valuation is drawn i.i.d.\ from
unif$(0,1)$. We generated $5$ instances for each point in the Cartesian product
of $\{2,4, 6, 8\}$ bidders and $\{4,6,\ldots,14\}$ goods. For each instance, we
computed (a) an \fppe as the solution to \mcp using the solver CVX~\citep{grant2008cvx}, and (b) an \sppe for every objective function using the MIP given by
\citet{conitzer2021pacing}. We considered at most $8$ bidders and
$14$ goods because of the limited scalability of the \sppe MIP; we were able to solve all \mcp instances in less than 2ms.

\subsection{Incentive Properties}
\label{sec:incentive_properties}

\paragraph{Analysis of Regret.} 
Similarly to what we have done for the realistic instances, first we look at the ex-post regret of each bidder under an \fppe compared to
being able to unilaterally deviate, while keeping the \fppe multipliers fixed
for all other bidders. We consider two sets of deviations: jointly setting bids through a best-response pacing multiplier, and the more powerful model of individually setting bids
in each auction. 

The resulting relative regrets are shown in
Figure~\ref{fig:app exp regrets}. The figure shows summary statistics over the maximum
relative ex-post regret, which is the fraction of the truthful-response value
that the bidder improves by if they deviate. For each instance, for each bidder,
we compute the optimal best response, subject to budget constraints.
The max is over bidders in the auction, and the statistic is
across instances. The middle line on each box is the median; the lower and upper
hinges show the first and third quartiles. The line extending from the box shows
outliers within 1.5 times the inter-quartile range. 

Compared to the results on realistic data, the regret experienced by bidders is higher, but qualitatively we see similar effects: the regret decreases when market density increases. In particular for 6, 8, and 10 bidders, the median regret for alternative multipliers is 0. Strikingly, the max relative regret under the model where bidders can set individuals bids does very poorly for 2 and 4 bidder instances. This is likely due to the expected difference between the first and second order statistic of $n$ draws from a uniform distribution (which shrinks as $O(\tfrac1n)$) which is the difference between the highest and second highest bid.

\begin{figure}[t]
  \centering
    \includegraphics[width=0.49\columnwidth]{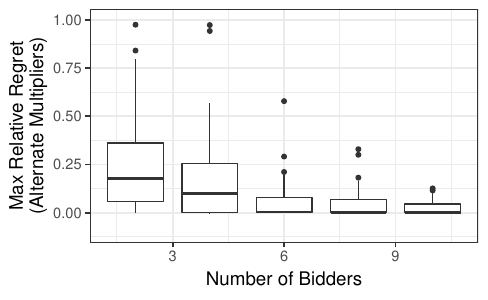}
	\includegraphics[width=0.49\columnwidth]{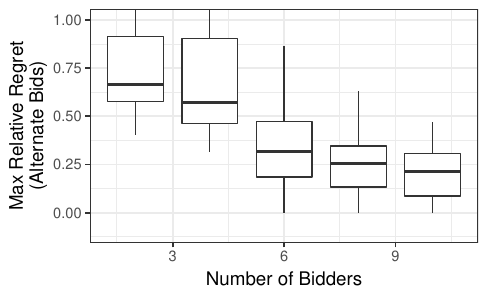}
	\caption{Summary statistics over the maximum relative ex-post regret (maximum taken over
    bidders in a given auction, the statistic taken over the maximum regret across
    instances). {\em Left\/}: The bidder can choose a single alternative pacing multiplier. {\em Right\/}: The bidder can choose an individual bid for every auction.}
  \label{fig:app exp regrets}
\end{figure}

\paragraph{Misreporting to Influence Equilibria.}
 
Figure~\ref{fig:app exp misreport} shows the resulting regrets when a single bidder can misreport by scaling their valuations and budgets with factors $(\lambda_v, \lambda_b) \in [0.1,0.2,\ldots,1.1]^2$ so the inputs provided to the proxy bidder are the true figures times those scalars. Similar to the realistic instances, the regret is generally lower than for alternative pacing multipliers.

\begin{figure}[t]
\centering
\includegraphics[width=0.49\columnwidth]{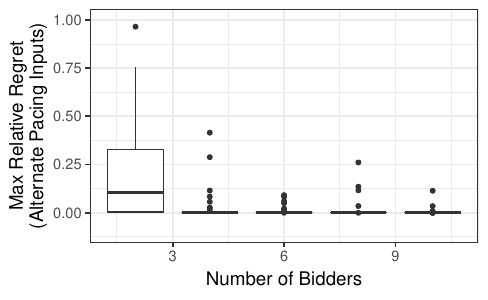}
 \vspace{-4mm}
\caption{Summary statistics when a bidder can misreport their values and budget to influence the \fppe.}
 \vspace{-2mm}
 \label{fig:app exp misreport}
\end{figure}

\paragraph{Regret as a Function of the Tightness of Budget Constraints.} 

We consider the effect of the pacing multiplier on regret. To generate instances with larger pacing multipliers, we compute \fppe on a modified set of problem instances, which are identical to instances in the preceding experiments, except that all budgets were scaled up by a factor of two. We show the relative regret for these instances in Figure~\ref{fig:app multiplier vs bid 2x}. Each point represents a bidder in one of the instances. The uptick in regret is slightly higher than for realistic data, though qualitatively similar.

\begin{figure}[t]
  \centering 
	\includegraphics[width=0.49\columnwidth]{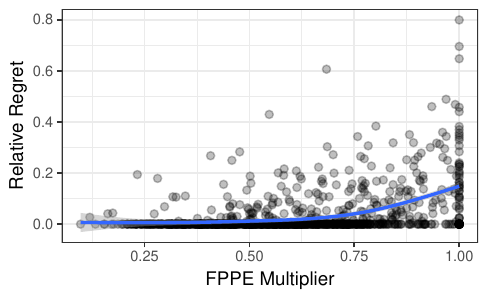}
  \caption{Relative regret as a function of \fppe pacing multiplier, run on instances with budgets scaled by a factor of 2.}
  \label{fig:app multiplier vs bid 2x}
\end{figure}

\subsection{Revenue and Welfare}\label{app:revenue_welfare}

In this section we compare revenue and social welfare under \fppe and \sppe on an instance-by-instance basis. The comparison is shown in Figure~\ref{fig:app exp sppe comparison}. The left panel shows the cdf
of the distribution of ratios of \fppe and \sppe revenues, while the right panel does the same for the ratios of welfare. Putting in perspective the curves corresponding to realistic instances shown in Figure~\ref{fig:exp sppe comparison}, we see that \fppe
revenue is now the same for about
$75\%$ of instances (compared to $30\%$ before), and almost never more than $1.5$ times as high (compared to $4.5$ times before). For welfare we observe that the CDS of the ratio between \fppe and \sppe look exactly the same compared to the realistic instances, albeit the x-axis representing the ratios are at a different scale.

\begin{figure}[t]
  \centering
	\includegraphics[width=0.49\columnwidth]{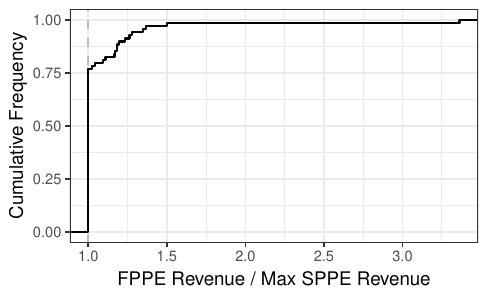}
	\includegraphics[width=0.49\columnwidth]{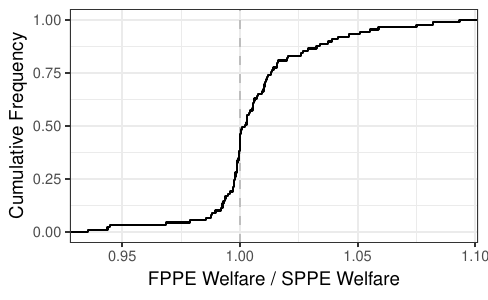}
  \caption{Distribution of the ratio of \fppe to \sppe revenues (left) and social welfare (right). }
  \label{fig:app exp sppe comparison}
\end{figure}

\end{APPENDICES}

\ACKNOWLEDGMENT{

We would like to thank 
Santiago Balseiro, 
Ozan Candogan, 
Roberto Cominetti,
Nikhil Devanur,
Yoni Gur, 
Renato Paes Leme,
Marco Scarsini,
and 
Marc Schroder
for very useful discussions on the technical content of this paper.
In addition, we extend our sincere thanks to 
the review team of Management Science,
conference PC members and reviewers at EC, 
seminar participants at Duke CS-Econ, Michigan Ross, MIT ORC and Stanford GSB, 
and attendees of 
the AAMAS-IJCAI workshop on Agents and Incentives in Artificial Intelligence (AI$\hat\ $3),
the AI and Marketing Science workshop at AAAI,
the Algorithmic Game Theory and Data Science workshop at EC,
the Barbados AGT workshop,
the International Conference on Game Theory,
and 
the Schloss Dagstuhl seminar on Traffic Models.
Their remarks and questions allowed us to significantly improve the manuscript and its presentation.
Debmalya Panigrahi was supported in part by NSF Awards CCF 1535972, CCF 1750140, and CCF 1955703.

}


\bibliographystyle{informs2014} 
\bibliography{refs}

\end{document}